\documentclass[12pt]{iopart}

\usepackage{iopams}  

\expandafter\let\csname eqalign\endcsname\relax
\expandafter\let\csname equation*\endcsname\relax
\expandafter\let\csname endequation*\endcsname\relax 
\expandafter\let\csname bs\endcsname\relax

\usepackage{amssymb,amsmath,mathrsfs,stmaryrd,amsthm,mathtools,amsfonts,graphicx,hyperref,url,booktabs,nicefrac,comment,microtype,tikz-cd,lingmacros,nccmath,tree-dvips, bbm, bm, etoolbox,algorithm,caption,subcaption,color,enumitem, physics, lineno, setspace} 

\DeclareMathOperator*{\argmax}{argmax}

\newcommand{\bs}[1]{\bm{#1}} 

\newtheorem{theorem}{Theorem}[section]
\newtheorem{thm-defn}[theorem]{Theorem/Definition}

\newtheorem{prop}[theorem]{Proposition}

\theoremstyle{definition}
\newtheorem{definition}[theorem]{Definition}

\theoremstyle{remark}
\newtheorem{remark}[theorem]{Remark}

\begin{document}

\title[A new measure between sets of probability distributions]{A new measure between sets of probability distributions with applications to erratic financial behavior} 

\author{Nick James$^1$ and Max Menzies$^2$\footnote{Equal contribution}}

\address{$^1$ School of Mathematics and Statistics, University of Melbourne, Victoria 3010, Australia}
\address{$^2$ Beijing Institute of Mathematical Sciences and Applications, 
Tsinghua University, Beijing 101408, China}
\ead{nick.james@unimelb.edu.au}

\vspace{10pt}
\begin{indented}
\item[]December 6 2021
\end{indented}

\begin{abstract}
This paper introduces a new framework to quantify distance between finite sets with uncertainty present, where probability distributions determine the locations of individual elements. Combining this with a Bayesian change point detection algorithm, we produce a new measure of similarity between time series with respect to their structural breaks. First, we demonstrate the algorithm's effectiveness on a collection of piecewise autoregressive processes. Next, we apply this to financial data to study the erratic behavior profiles of 19 countries and 11 sectors over the past 20 years. Our measure provides quantitative evidence that there is greater collective similarity among sectors' erratic behavior profiles than those of countries, which we observe upon individual inspection of these time series. Our measure could be used as a new framework or complementary tool for investors seeking to make asset allocation decisions for financial portfolios.

\end{abstract}


\vspace{2pc} 
\noindent{\it Keywords}: Time series analysis, change point detection, Bayesian analysis, similarity measures, market crises

\submitto{\JSTAT} 


\section{Introduction}
\normalsize

In this paper, we marry two widely researched areas within statistics and the natural sciences: similarity and anomaly detection. Our primary means of detecting similarity that we build upon is the use of distance measures such as metrics; our primary means of anomaly detection is change point detection, incorporating a Bayesian perspective around change points. We introduce a new framework of \emph{sets with uncertainty} (which are sets of probability distributions) and a new family of \emph{semi-metrics} between them. Then, we apply this to measure the similarity between time series' sets of change points, as proxies for their erratic behavior profiles.

Metric spaces appear throughout mathematics. One particular subfield that has seen substantial recent activity is the study of metrics between sets, specifically subsets of an ambient metric space. The most utilized metric in this context is the \emph{Hausdorff metric}, which we introduce and summarize in Section \ref{sec:review}. Distances between sets have proven useful in many applications, including image detection and matching \cite{Gardner2014, Dubuisson1994, Rucklidge1995, Rucklidge1996, Rucklidge1997}, the study of fuzzy sets \cite{Brass2002, Rosenfeld1985, chaudhuri1996, Chaudhuri1999, Boxer1997, Fan1998}, and efficient computational geometry \cite{Rote1991,Li2008,Eiter1997, Atallah1983, Atallah1991, Barton2010, Shonkwiler1989, Huttenlocher1990, Huttenlocher1992, Aspert2002}. The primary challenge in this area is that the Hausdorff and other metrics \cite{Fujita2013} are highly sensitive to outliers \cite{Baddeley1992}, while alternative \emph{semi-metrics} may not satisfy the triangle inequality property of a metric. Relevant definitions and recent advances will be reviewed in Section \ref{sec:review}.

In recent years, there has been substantial interest in the use of distance measures between time series \cite{Moeckel1997}, for example to understand similarity in movement between financial assets \cite{Dose2005, Basalto2007, Basalto2008}. More recent work has prioritized distances between time series based on certain critical points, requiring the use of similarity measures between finite sets. These critical points often carry particular importance, capturing information about the broader behavior of a time series. For example, \cite{james2020covidusa} studied turning points in COVID-19 case time series, which may summarize undulating wave behavior and separate the data into different waves of the disease. Outside time series, metric learning has become a popular topic within the field of computational statistics and machine learning more broadly, where a distance function is optimally tuned for a candidate task. Various applications include computer vision \cite{Hua2007, Snavely2006}, text analysis \cite{Davis2008, Lebanon2006} and program analysis \cite{Ha2007}.

A natural corollary of the use of distance measures between time series (and more broadly, any sort of data) is the detection of anomalies. Anomaly detection is a well-researched problem across many data sets and spaces, incorporating various techniques from statistics and machine learning \cite{Breunig2000,Liu2012,Amer2013,chalapathy2019deep,Pang2021}. Our paper follows a line of recent work that aims to detect entire time series, not just single data points, as anomalous. Several means to define anomalous time series may be used, including geometric measures, principal components, and shapelet transformations \cite{Hyndman2015,Beggel2018,cochrane2020anomaly}. Our primary means of observing anomalous time series is agglomerative hierarchical clustering on our new distance measure.

Change point detection is an important subfield of time series analysis. Developed by Hawkins et al. \cite{Hawkins1977,Hawkins2005}, change point (or structural break) detection algorithms estimate the points in time at which the stochastic properties of a time series, such as its underlying mean or variance, change. Traditionally, most algorithms apply hypothesis testing frameworks \cite{Pettitt1979,Ross2011T,Ross2012,Ross2013,Ross2013physa,Ross2014}, which do not quantify the uncertainty surrounding the change points and typically require the assumption of independent data. In a financial setting, this assumption is inappropriate, as rich patterns in correlation structure have been observed in the literature \cite{Fenn2011, Laloux1999, Mnnix2012, Heckens2020, Wilcox2007, Driessen2003, Ausloos2000}. Bayesian methodologies have been introduced to ameliorate both these issues, providing uncertainty intervals around candidate change points in a setting where time series may exhibit strong autocorrelation \cite{Rosen2009,Rosen2012}.

Recent work combined hypothesis testing change point algorithms with semi-metrics between finite sets to produce new distance measures between time series \cite{James2020_nsm}, and applied this to measure similarity in erratic behavior \cite{James2021_crypto}. However, it did not consider the uncertainty inherently present in change point detection, and we are not aware of any existing work that computes distances between sets, where there are uncertainty intervals associated with individual elements in the set. Thus, we introduce a new framework of \emph{sets with uncertainty} (which are certain sets of probability density functions), where probability distributions determine the uncertainty surrounding various elements' locations, and appropriate semi-metrics between them. Then, we combine this new family of semi-metrics with a Bayesian change point detection algorithm, which records the uncertainty around change points, and apply this to quantify distance between time series' sets of change points with uncertainty. This is presented in Section \ref{sec:MJWtheory}.

We apply our procedure to a financial context in Section \ref{sec:results}. Analyzing both countries and sectors, we reveal similarity and anomalies in the long term dynamics of various indices with respect to change points in their behavior. We also demonstrate that our methodology may fit well with more conventional statistical or time series analysis. In our case, we complement our analysis with a closer examination of some of the most frequently observed change points across countries and sectors, associated to the global financial crisis (GFC) and COVID-19. We discuss insights and implications from both our utilized new and existing methodologies in Section \ref{sec:conclusion}.

\section{Review of existing (semi)-metrics}
\label{sec:review}

In this section, we review some properties of a metric space and the existing (semi)-metrics we draw upon, including the Wasserstein metric between probability distributions, the Hausdorff metric between sets, and recently introduced semi-metrics between finite sets.

\begin{definition}[Metric space]
\label{defn:metric}
Let $X$ be a set and $d: X \times X \to \mathbb{R}_{\geq 0}$ a function. Suppose $d$ satisfies the following properties for all $x,y,z \in X$:
\begin{enumerate}
    \item $d(x,y)=0$ if and only if $x=y$;
    \item $d(x,y) = d(y,x)$;
    \item $d(x,z) \leq d(x,y)+d(y,z)$.
\end{enumerate}
Then we say $d$ is a \emph{metric} and the pair $(X,d)$ is a \emph{metric space} \cite{RudinPMI}.

\noindent Alternatively, suppose $d: X \times X \to \mathbb{R}_{\geq 0}$ satisfies only conditions (i) and (ii). Then we say $d$ is a \emph{semi-metric} on $X$.
\end{definition}
Condition (iii) is known as the \emph{triangle inequality}. A particularly important metric is the \emph{Wasserstein metric}, which is used as a distance between two probability measures.

\begin{definition}[Wasserstein metric]
Let $(X,d)$ be a complete and separable metric space. Suppose $\mu$ and $\nu$ are Borel probability measures on $X$, and $q \geq 1$ \cite{Clement2007}. Let $\Gamma(\mu,\nu)$ be the set of all Borel probability measures on $X \times X$ with marginal distributions $\mu$ and $\nu$ respectively. The \emph{Wasserstein metric} \cite{Clement2007} is defined as follows:
\begin{equation}
    W_{q} (\mu,\nu) = \inf_{\gamma \in \Gamma(\mu,\nu)} \bigg( \int_{X \times X} d^{q} (x,y) d\gamma  \bigg)^{\frac{1}{q}}.
\end{equation}
\end{definition}
The Wasserstein distance is most commonly applied to Euclidean space $X=\mathbb{R}^m$ with $d$ as the standard metric. On $X=\mathbb{R}$, the computation of the Wasserstein metric simplifies considerably. If probability measures $\mu, \nu$ on $\mathbb{R}$ have associated cumulative distribution functions (cdf's) $F,G$ and quantile functions \cite{Gilchrist2000} $F^{-1}, G^{-1}$ respectively, the Wasserstein metric can be computed \cite{DelBarrio} as
\begin{align}
\label{eq:Wassequation}
   \left( \int_{0}^1 |F^{-1} - G^{-1}|^q dx\right)^{\frac{1}{q}}.
\end{align}
In this paper, we shall only compute the Wasserstein metric between discrete probability density functions valued on the real line, which always have associated cdf's and quantile functions.

The other class of (semi-)metrics we will use are those between subsets of a given metric space. We begin with a preliminary definition. Let $S$ be a subset of a metric space $X$, and $x$ an element of $X$. Then the distance from the element to the set is defined as the minimal distance from $x$ to $S$, computed as follows:
\begin{equation}
d(x,S) = \inf_{s \in S} d(x,s). \label{min distance defn}
\end{equation}
When the metric $d$ is the Wasserstein metric $W_q$, we denote this minimal distance $d_W$, suppressing $q$ from the notation.

\begin{definition}[Hausdorff metric]
Let $(X,d)$ be a metric space. Suppose $S$ and $T$ are closed and bounded subsets of $X$. The \emph{Hausdorff metric} \cite{Conci2017} between $S$ and $T$ is defined as follows:
\begin{align}
    d_{H}(S,T) =&  \text{max } \bigg( \sup_{s \in S} d(s,T), \sup_{t \in T} d(t,S) \bigg), \\
    =&  \sup \{ d(s,T), s \in S; d(t,S), t \in T \}. 
\end{align}
\end{definition}
This is the supremum or $L^{\infty}$ norm of all minimal distances from elements $s \in S$ to $T$ and $t\in T$  to $S$. The Hausdorff metric is a metric on the set of all closed bounded subsets of $X$. The Hausdorff metric satisfies the triangle inequality, but is highly sensitive to even a single outlier. In \cite{James2020_nsm}, we introduced the following semi-metric, which replaces the $L^{\infty}$ norm with an $L^p$ average ($p \geq 1$). It can only be defined for finite sets $S,T.$
\begin{definition} 
\label{New MJ distance}
Let $(X,d)$ be a metric space, and $p \geq 1$. Suppose $S$ and $T$ are finite subsets of $X$. The MJ$_p$ distance \cite{James2020_nsm} between $S$ and $T$ is defined as follows:
\begin{equation}
\label{eq:MJ}
    d^p_{MJ}({S},{T}) = \Bigg(\frac{\sum_{t\in T} d(t,S)^p}{2|T|} + \frac{\sum_{{s} \in {S}} d(s,T)^p}{2|S|} \Bigg)^{\frac{1}{p}}.
\end{equation}
\end{definition}
As $p \to \infty$, this family of semi-metrics includes the Hausdorff metric as its limit element. In particular, changing $p$ may provide a trade-off between reduced sensitivity to outliers (better when $p$ is small) and greater satisfaction of the triangle inequality (better when $p$ is large). Several theoretical and experimental properties of these semi-metrics are presented in \cite{James2020_nsm}.

The primary application of the semi-metrics in that work was between sets of time series' change points. Specifically, given a collection of real-valued time series $(X^{(i)}_t), i=1,...,n$, a change point detection algorithm was applied that produced sets of structural breaks $S_i$ for each time series, $i=1,...,n$. Then, the distance between time series was defined as $d^p_{MJ}(S_i,S_j)$. However, this work did not take into account any potential uncertainty in the change points.
 
\section{MJ-Wasserstein framework and methodology}
\label{sec:MJWtheory}
In this section, we introduce our framework of \emph{sets with uncertainty}, define appropriate distances between such objects, and describe the primary application of this manuscript, namely quantifying similarity between sets of change points with uncertainty present.

\begin{definition}[Set with uncertainty] 
Let $f_1,...,f_k$ be probability density functions on $\mathbb{R}$. Suppose their supports are disjoint intervals $I_1,...,I_k$. Then we call the set $\{ f_1,...,f_k\}$ a \emph{set with uncertainty} and $k$ its size.
\end{definition}
That is, a set with uncertainty allows the positions of each element to vary according to specified probability distributions. The disjointness of the intervals is analogous to the requirement that the elements of a regular set be distinct. We will use the notation $\tilde{S}, \tilde{T}$ for sets with uncertainty and $S,T$ for regular sets. In the specific case when the probability densities are just delta functions supposed at single points, that is, $\tilde{S}=\{\delta_{x_1},...,\delta_{x_k}\}$, we can associate a regular set $S=\{x_1,...,x_k\}$, and vice versa.

\begin{definition}[MJ-Wasserstein semi-metric]
Let $\tilde{S}=\{f_1,...,f_k\}$ and $\tilde{T}=\{g_1,...,g_l\}$ be two sets with uncertainty, and $p \geq 1$. The MJ-Wasserstein semi-metric $d^p_{MJW}(\tilde{S},\tilde{T})$ is defined as follows:
\begin{align}
\label{eq:MJW}
    d^p_{MJW}(\tilde{S},\tilde{T}) &= \Bigg(\frac{\sum_{g\in \tilde{T}} d_W(g,\tilde{S})^p}{2|\tilde{T}|} + \frac{\sum_{{f} \in \tilde{S}} d_W(f,\tilde{T})^p}{2|\tilde{S}|} \Bigg)^{\frac{1}{p}} \\ &= \Bigg(\frac{\sum_{j=1}^l d_W(g_j,\tilde{S})^p}{2l} + \frac{\sum_{i=1}^k d_W(f_i,\tilde{T})^p}{2k} \Bigg)^{\frac{1}{p}}.
\end{align}
\end{definition}
That is, for each distribution $f_i$, its closest distribution $g_j$ in $\tilde{T}$ is found according to the Wasserstein metric, and this minimal distance is computed. Essentially, this combines the MJ$_p$ semi-metric on a general metric space with the Wasserstein metric on probability distributions.

\begin{remark}
In this remark, we briefly explain why the MJ-Wasserstein semi-metric can be viewed as a direct generalization of the MJ$_p$ defined in (\ref{eq:MJ}). Indeed, let $\tilde{S},\tilde{T}$ be two sets with uncertainty consisting only of Dirac delta functions. To these we can associate regular sets $S_0,T_0$ as discussed above. By (\ref{eq:MJ}), (\ref{eq:MJW}) and the fact that $W_q(\delta_a, \delta_b)=|a-b|$, it follows that $d^p_{MJW}(\tilde{S},\tilde{T})=d^p_{MJ}({S_0},{T_0})$. That is, the MJ-Wasserstein distance between $\tilde{S},\tilde{T}$ reduces to the existing MJ$_p$ distance between $S_0,T_0$.

\end{remark}

\begin{remark}
This property is by no means the only reason we select the Wasserstein metric for our methodology. The well-known Kullback-Leibler divergence \cite{Kullback1951} is unsuitable in our context because it returns a value of infinity when two probability distributions have disjoint support \cite{Cover2006}, as will frequently be the case (where time series' change points are not located close to each other). Similarly, the Jensen-Shannon metric \cite{Endres2003} always returns its greatest possible value of 1 (depending on normalization convention \cite{Lin1991}) between two densities of disjoint support, so is usually uninformative in our application.

On the other hand, the Wasserstein metric is informative regarding where the change points occur, taking uncertainty into account. This is crucial in our application, as we want to keep track of the positions in time when time series behaviors change, and which time series change at similar points in time. Finally, we select the Wasserstein over the related Radon metric, as the latter lacks desirable analytical properties, such as sequential compactness \cite{Daletskii1985}.

\end{remark}

\begin{prop}
\label{prop:triangle inequality again}
For $p >0$, the MJ-Wasserstein distance measures are semi-metrics. However, they fail the triangle inequality up to any constant. That is, there is no constant $k$ such that
\begin{align} \label{TriIn}
    d^p_{MJW}(\tilde{S},\tilde{R}) \leq k(d^p_{MJ}(\tilde{S},\tilde{T})+d^p_{MJ}(\tilde{T},\tilde{R}))
\end{align}
for any sets with uncertainty $\tilde{S},\tilde{T},\tilde{R}$.
\end{prop}
\begin{proof}
First, (\ref{eq:MJW}) is defined symmetrically, so we clearly see that $d^p_{MJW}(\tilde{S},\tilde{T})=d^p_{MJW}(\tilde{T},\tilde{S})$. Next, if $d^p_{MJW}(\tilde{S},\tilde{T})=0$, this forces every term $d_W(g, \tilde{S})=0$ for each $g \in \tilde{T}$. As the Wasserstein distance is a metric, this implies $g \in \tilde{S}$. Reversing the reasoning, we also see that each probability density function in $\tilde{S}$ must lie in $\tilde{T}$. So the sets with uncertainty $\tilde{S}$ and $\tilde{T}$ must be equal, with the exact same probability distributions as members. This establishes properties (i) and (ii) of Definition \ref{defn:metric}.

Finally, we demonstrate the failure of the triangle inequality up to any constant. We can find an appropriate counterexample when $\tilde{S},\tilde{T},\tilde{R}$ are regular sets $S,T,R$. We modify the example of failure of the triangle inequality for the MJ semi-metric, presented in \cite{James2020_nsm}, Proposition 3.6. Specifically, suppose $a,b$ are two elements in an an ambient metric space $X$ with $d(a,b)=d$. Let $S=\{a,b\}$ and $R=\{b \}$. Next, suppose $b_1,\ldots,b_n$ are all within a distance $\epsilon$ of $b$. Let $T=\{a,b,b_1,\ldots,b_{n-2}\}$. Essentially, $T$ has $n$ elements, with $n-1$ of them closely bunched together. Then

\begin{eqnarray*}
d^p_{MJW}(\tilde{S},\tilde{T})=d^p_{MJ}(S,T) & \leq & \Big( \frac{1}{2n} (n-2)\epsilon^p \Big)^{\frac{1}{p}} \leq \epsilon; \\
d^p_{MJW}(\tilde{T},\tilde{R})=d^p_{MJ}(T,R)& \leq & \left[ \frac{1}{2n}\Big( d^p+(n-2)\epsilon^p \Big) \right]^{\frac{1}{p}} \leq \Big(\frac{d^p}{2n} + \epsilon^p \Big)^{\frac{1}{p}};\\
d^p_{MJW}(\tilde{S},\tilde{R})=d^p_{MJ}(S,R)&=&\Big( \frac{1}{4}d^p \Big)^{\frac{1}{p}} = 4^{\frac{-1}{p}}d.
\end{eqnarray*}
Suppose $\epsilon$ is chosen such that $\epsilon<d(2n)^{\frac{-1}{p}}<dn^{\frac{-1}{p}}$. Then $d^p_{MJW}(\tilde{S},\tilde{T})+d^p_{MJW}(\tilde{T},\tilde{R}) \leq \frac{2d}{n^{\frac{1}{p}}}$

\noindent Carefully noting what is above, $d^p_{MJW}(\tilde{S},\tilde{T})+d^p_{MJW}(\tilde{T},\tilde{R})= O(d n^{\frac{-1}{p}})$ while $d^p_{MJW}(\tilde{S},\tilde{R}) =\Theta(d)$. Choosing $n$ sufficiently large, with $\epsilon<d(2n)^{\frac{-1}{p}}$, we deduce there is no universal modified triangle inequality for the MJ-Wasserstein distance.

\end{proof}

\begin{remark}
The primary cause for failure of the triangle inequality, where $d^p_{MJW}(\tilde{S},\tilde{T})+d^p_{MJW}(\tilde{T},\tilde{R})$ may be large compared to $d^p_{MJW}(\tilde{S},\tilde{R})$, is a large size of $\tilde{T}$. To illustrate this, let the \emph{asymmetric distance} $d^a_\to(\tilde{S},\tilde{T})$ be defined as follows:
\begin{align}
d^a_\to(\tilde{S},\tilde{T})=\frac{\sum_{f\in \tilde{S}} d_W(f,\tilde{T})^p}{2|\tilde{S}|}.
\end{align}
This measures a one-way distance from $\tilde{S}$ to $\tilde{T}$. Then $d^p_{MJW}(\tilde{S},\tilde{T})=\left(d^a_\to(\tilde{S},\tilde{T})+ d^a_\to(\tilde{T},\tilde{S})\right)^\frac{1}{p}$. Suppose that $\tilde{T}$ contains $\tilde{S}$ and $\tilde{R}$ as subsets. Then $d^a_\to(\tilde{S},\tilde{T})= d^a_\to(\tilde{R},\tilde{T})=0$. Thus
\begin{align}
\label{eq:oneside}
    d^p_{MJW}(\tilde{S},\tilde{T})+d^p_{MJW}(\tilde{T},\tilde{R})= d^a_\to(\tilde{T},\tilde{S})^\frac{1}{p} + d^a_\to(\tilde{T},\tilde{R})^\frac{1}{p}
\end{align}
Then, increasing the size of $\tilde{T}$ with bunched elements may reduce (\ref{eq:oneside}) relative to  $d^p_{MJW}(\tilde{S},\tilde{R})$. That is, the triangle inequality may be violated when $\tilde{T}$ is excessively ``large'' relative to $\tilde{S}$ and $\tilde{R}$, for example if $\tilde{T}$ contains $\tilde{S}$ and $\tilde{R}$ as subsets as well as bunched or irrelevant elements.

\end{remark}

\begin{remark}
For this reason, we require sets with uncertainty, by definition, to have disjoint intervals of support. Not only is this always the case in the output of our change point algorithm, it also helps to reduce the effects of ``bunching,'' which may cause a failure of the triangle inequality. The choice of the Wasserstein metric is useful here too: disjoint supports of the constituent probability distributions in $\tilde{S}=\{f_1,...,f_k\}$ prevent bunching in the values of $d_W(f_i,f_j)$; this was the cause of the violation of the triangle inequality in the proof of Proposition \ref{prop:triangle inequality again}. Even more simply, the disjointness condition prohibits duplication (or near-duplication) of elements of $\tilde{S}$. (Near-)duplication of elements may also throw off the triangle inequality.

\end{remark}

\begin{remark}
We briefly compare our framework sets with uncertainty to the existing notion of fuzzy sets. A fuzzy set is a pair $A=(X,m)$ where $X$ is a reference set and $m: X \to [0,1]$ is a membership function. The membership function describes whether elements $x \in X$ are not included in the fuzzy set $(m(x)=0)$, partially included $(0<m(x)<1)$ or fully included $(m(x)=1)$. Fuzzy sets are frequently used in image detection and pattern recognition \cite{Brass2002, Rosenfeld1985, chaudhuri1996, Chaudhuri1999, Boxer1997, Fan1998}. Fuzzy sets are typically not considered probabilistic in nature; there is no requirement for values of $m(x)$ to add to 1 over particular intervals, and the membership function does not model uncertainty in the location of an element, only degree of membership.

In addition, existing distance measures between fuzzy sets are based on measures between sets, not probability density functions. Specifically, (semi)-metrics between fuzzy sets use the Hausdorff distance \cite{Rosenfeld1985}, typically between sets $A^{\geq r} = \{x \in X: m(x) \geq r\}$. That is, our semi-metrics between sets with uncertainty are of a different nature, prioritizing the probability distributions under which the elements vary, rather than using the Hausdorff metric between specified subsets of the reference set $X$.

There is a method to associate a fuzzy set to a set with uncertainty. If $\tilde{S}=\{f_1,...,f_k\}$ where probability distributions have supports $I_1,...,I_k$, we may set $X=I_1 \cup ... \cup I_k$ and $m(x)=f_j(x)$ if $x \in I_j$, zero otherwise. As the supports are assumed to be disjoint, this defines a fuzzy set. However, no existing measures between fuzzy sets produce the MJ-Wasserstein distance measure between sets with uncertainty. Thus, the sets with uncertainty framework is necessary to construct the new distance measures and the existing fuzzy sets category is not sufficient.

\end{remark}

\subsection{Application to time series}
\label{sec:applicationtodistances}
Now, we describe our application to time series' sets of change points with uncertainty. Let $(X^{(i)}_t), i=1,...,n$ be a collection of $n$ time series over a period $t=1,...,T$. Seeking to determine similarity between time series' erratic behavior profiles, taking into account uncertainty, we apply the Bayesian change point detection algorithm of Rosen et al. \cite{Rosen2012,Rosen2017} to each time series. Initially, this produces a distribution over the number of change points $m$, and conditional on $m$, a set of $m$ points with uncertainty intervals. We select the maximally likely number of change points $m^{(i)}_0$ for each individual time series, resulting in a set with uncertainty $\tilde{S}_i$ with $m^{(i)}_0$ elements. Finally, we define our distance between the $n$ time series as $D_{ij}=\frac{1}{T}d^1_{MJW}(\tilde{S}_i,\tilde{S}_j)$, setting $p=1$ in the MJ-Wasserstein parameter and $q=1$ in the Wasserstein metric. We remark that this is an abuse of notation, and that the distance can only be defined after a chosen change point detection algorithm has been performed. Nonetheless, the application of this entire methodology provides a useful framework for quantifying affinity between an attribute of time series behavior that is notoriously hard to capture, while considering uncertainty. We normalize by the length of the time series $T$ so we can compare features of collections of time series over different period lengths, which will be required in Section \ref{sec:results}. We provide full details of the change point algorithm in \ref{appendix:RJMCMCsampling}.

\subsection{Empirical analysis of the triangle inequality}
\label{sec:transitivity analysis}
As established in Proposition \ref{prop:triangle inequality again}, the MJ-Wasserstein semi-metric does not satisfy the triangle inequality (condition (iii) of Definition \ref{defn:metric}), even up to a constant. This could present a practical problem, as it could mean the property of sets with uncertainty being close under the semi-metric may not be transitive. That is, perhaps $d^p_{MJW}(\tilde{S}, \tilde{T})$ and $d^p_{MJW}(\tilde{T}, \tilde{R})$ could be sufficiently small, indicating substantial similarity between $\tilde{S}, \tilde{T}$ and $\tilde{T}, \tilde{R}$, but $d^p_{MJW}(\tilde{S}, \tilde{R})$ might not be small. Thus, we include an empirical analysis of the failure of the triangle inequality property with our use of the new semi-metric. Specifically, we examine two questions:
\begin{enumerate}
    \item how often does the $d^p_{MJW}$ semi-metric  fail the triangle inequality, and
    \item how badly do these quantities violate the triangle inequality?
\end{enumerate}
To explore this empirically, we generate a three-dimensional matrix and test whether the triangle inequality is satisfied for all possible triples of elements within the matrix. The matrix is defined as
\begin{equation}
\label{eq:transitivityanalysis}
    T_{i,j,k}=
    \begin{cases}
      \text{blue }, & \frac{D_{ik}}{D_{ij} + D_{jk}} \leq 1, \\
      \text{yellow },  & 1 < \frac{D_{ik}}{D_{ij} + D_{jk}} \leq 2, \\
      \text{red}, & \text{else,}
    \end{cases}
\end{equation}
where $D_{ij}$ give the distances between a collection of sets with uncertainty, as defined in Section \ref{sec:applicationtodistances}. We then record the proportion of triples $(i,j,k)$ that fail the triangle inequality as well as the mean value of $\frac{D_{ik}}{D_{ij} + D_{jk}}$ among the failed triples.

\subsection{Running cost of procedure}
\label{sec:runningcost}

In this section, we calculate the cost of our entire procedure, both the change point detection algorithm in \ref{appendix:RJMCMCsampling} and our subsequent calculation of the MJ-Wasserstein distances. As in Section \ref{sec:applicationtodistances}, suppose we have $n$ time series $(X^{(i)}_t)$ over a period $t=1,...,T$. Further, suppose as a result of implementing the change point algorithm, every considered number of change points $m$ is bounded above by a constant $M$. Finally, suppose that every constituent probability distribution within each $\tilde{S}_i$ is supported on at most $P$ elements. Alternatively, one may suppose the sizes of the support intervals $I^{(i)}_j$ are uniformly at most $P$.

First, the Bayesian change point detection algorithm proceeds via a reversible jump Markov chain Monte Carlo algorithm. We use a fixed number $K=10 000$ of iterations (including 5000 for burn in). The two most costly procedures within a single iteration are as follows. One is the search for optimal points $t$ within so-far-determined segments $[\xi_{k^{*}-1}, \xi_{k^{*}+1}]$. This has a cost of $O(T)$. The other is the inversion of a matrix of size at most $M \times M$. We reiterate that we assume our bound $M$ is greater than every possible considered number of change points $m$ within the sampling scheme of \ref{appendix:RJMCMCsampling}. So the inversion has a cost of at most $O(M^3)$ (but in practice is more efficient). Thus, the cost of the change point algorithm for a single time series is $O(K(M^3+T)$. It follows the total cost of implementing the change point algorithm across a collection of $n$ time series is $O(nKM^3+nKT))$. At this point, $n$ sets with uncertainty $\tilde{S}_i$, $i=1,...,n$ have been obtained.

Next, we consider the cost of computing a single MJ-Wasserstein distance $d^1_{MJW}(\tilde{S}_i,\tilde{S}_j)$. This computation involves at most $O(|S_i||S_j|)$ comparisons between probability densities, as seen in (\ref{eq:MJW}). By (\ref{eq:Wassequation}), computing each individual Wasserstein distance $d_W$ between probability distributions supported on intervals of length at most $P$ is of cost $O(P)$. Indeed, (\ref{eq:Wassequation}) calculates the difference between two step functions with at most $P$ steps. Hence, the computation of $d^1_{MJW}(\tilde{S}_i,\tilde{S}_j)$ is of cost at most $O(M^2P)$. Thus, computing our full matrix $D_{ij}$ between all pairs of time series is of cost $O(n^2M^2P)$.

Finally, with these matrix entries computed, the cost of the procedure to empirically analyze the triangle inequality is $O(n^3)$. Overall, the full cost of our procedure is $O(nKM^3+nKT+ n^2M^2P + n^3)$. In practice, $K$ and $T$ are much larger than any other parameter, so the total cost is $O(nKT)$, with the greatest cost occurring as a result of the sampling scheme.

\subsection{Validation of synthetic data}
\label{sec:validation}
In this section, we validate our methodology on six synthetic time series of length $T=1500$. Each is formed by concatenating a sequence of autoregressive processes whose parameters change at specified locations. Four of the generated time series are chosen to share similar locations of structural breaks, while two are quite different from the rest, exhibiting similarity only with each other. This will serve as a simplified representative of the real data in Section \ref{sec:results}, where there will frequently be a majority collection of similar time series with a smaller number of outlier sectors or countries.

The first four time series are chosen to have change points at approximately $t=200, 500, 700, 900, 1100, 1300$ while the latter two have change points at approximately $t=750$. We remark that the detection of change points in this context has traditionally been very difficult due to limited or no changes in the process mean or variance between autoregressive processes and the high amount of autocorrelation. Nonetheless, our methodology performs well at identifying change points in these synthetic time series, and our distance measure clearly identifies the similarity between the first four time series and the outlier status of final two. In Figure \ref{fig:ARdend}, we display hierarchical clustering on the obtained $6 \times 6$ matrix $D_{ij}$, which identifies the similarity of the first four time series and the outlier status of the final two. Here and elsewhere in this paper, we implement agglomerative hierarchical clustering via the average-linkage method \cite{Mllner2013}. In \ref{app:autoregressive}, we provide full details of these synthetic time series, including all piecewise autoregressive components, and show the change points detected.

Finally, we include a brief sensitivity analysis of the ability of our change point algorithm to detect small changes in adjacent autoregressive processes. Results are promising, indicating successful detection except for very small changes in the autoregressive parameters. This experiment is also detailed in \ref{app:autoregressive}.

\begin{figure}
    \centering
    \includegraphics[width=\textwidth]{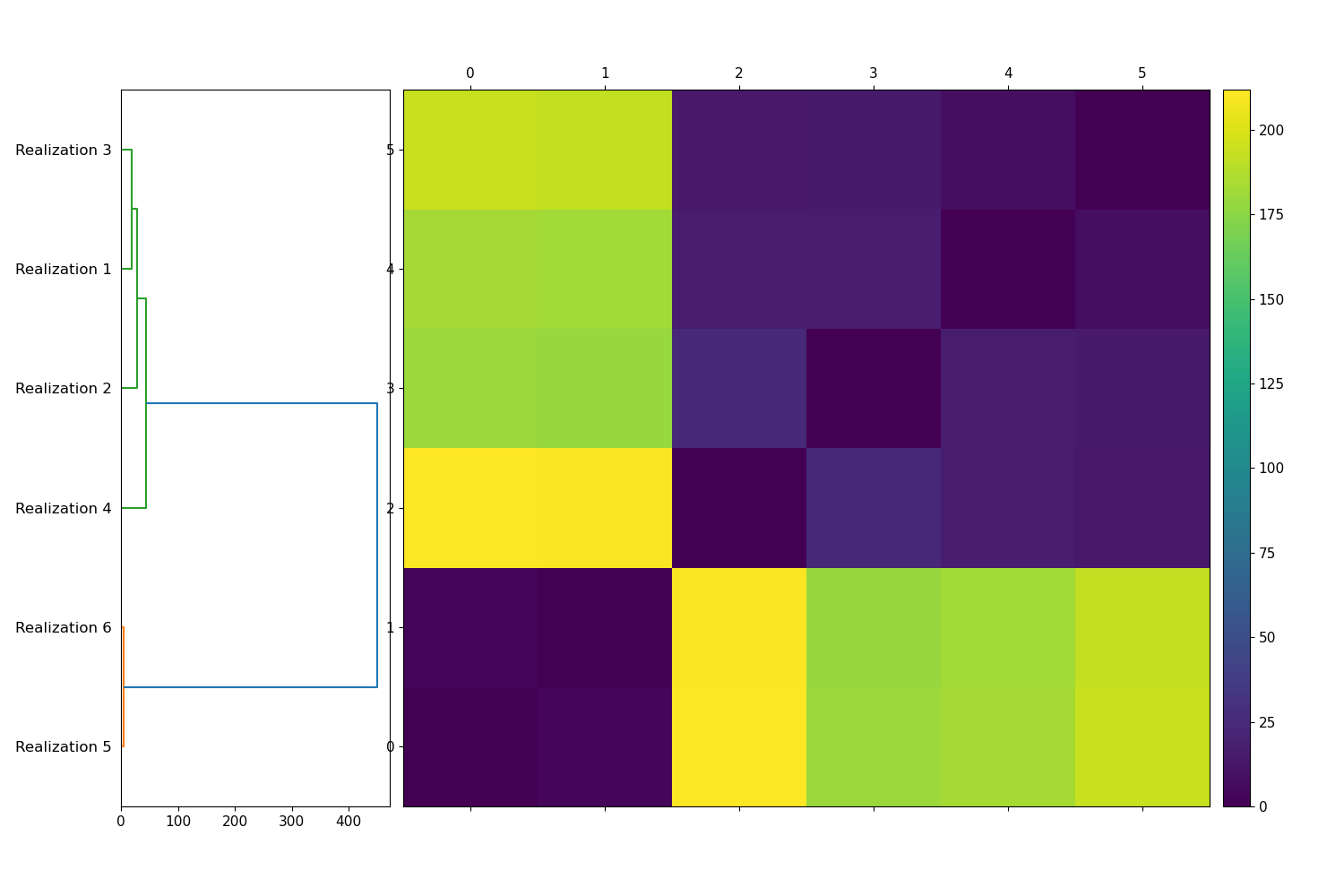}
    \caption{Hierarchical clustering applied to six synthetic time series obtained from piecewise autoregressive processes. Our semi-metric clearly identifies the strong similarity between the first four, and the outlier status of the final two.}
    \label{fig:ARdend}
\end{figure}

\section{Erratic behavior analysis of time series}
\label{sec:results}

In this section, we study the erratic behavior profiles of 19 country and 11 sector indices. Country data runs from 01/01/2002-10/10/2020, while sector data runs from 01/01/2000-10/10/2020. First, we generate a log return time series for each country and sector index. Next, we apply a Bayesian change point detection algorithm, as detailed in \ref{appendix:RJMCMCsampling}, to generate a set of change points with uncertainty for each index. Then, we may apply our semi-metrics between such sets with uncertainty, as defined in Section \ref{sec:MJWtheory}.

\subsection{Country similarity}
First, we apply the methodology of Section \ref{sec:MJWtheory} to the log return time series of 19 country indices. We produce a $19 \times 19$ distance matrix $D^c_{ij}$ that measures similarity between the erratic behavior profiles of these time series, as summarized by their sets of change points with uncertainties. We display hierarchical clustering of $D^c$ in Figure \ref{fig:Country_dendrogram}. The cluster structure of these indices reveals one predominant cluster of countries, and a minority collection of various outliers. The primary cluster consists of Australia, Canada, France, Germany, Italy, Korea, the Netherlands, Saudi Arabia, Switzerland, Turkey, the United Kingdom (UK) and the United States (US). The outlier countries include Brazil, China, India, Indonesia, Japan, Russia and Spain. The primary cluster contains a strong subcluster of mostly European composition, such as France, Italy, the Netherlands, Switzerland and the UK. This may be due to substantial similarity in the geographic location and political association between such countries, leading to similar change point propagation surrounding events such as Brexit. Notably, Spain is determined to be markedly less similar than its European counterparts. In this section and elsewhere, we refer to a time series (in this case pertaining to countries) as anomalous if it lies in its own cluster as determined by our implementation of agglomerative hierarchical clustering via the average-linkage method \cite{Mllner2013}. In this instance, the unique anomalous country within the collection is China. This is unsurprising, as many of the idiosyncratic restrictions related to their local equity market may provide different dynamics to those exhibited by other countries.

We take a closer look at the similarity of select country indices in Figure \ref{fig:CountryTimeSeries}. First, we observe clear similarity in the log return time series for France, Germany, and the UK, in Figures \ref{fig:France_index}, \ref{fig:Germany_index} and \ref{fig:UK_index}, respectively. All three of these time series exhibit change points that are, broadly speaking, evenly distributed over the entire period of analysis. Further, all three experience a change point in 2008-2009, associated to the global financial crisis (GFC), in mid-2016, associated with Brexit, and early 2020, associated with COVID-19. These three examples are also representative of the behavior of Australia, Canada, Italy, the Netherlands and Switzerland, all of which lie in the same subcluster of similarity. Next, we display the time series for the US (\ref{fig:US_index}); this is characterized by a long period of stability from 2012-2019, associated with a prolonged bull market. The strong similarity between France, Germany and the UK and the slightly less affinity between them and the US are all visible in Figure \ref{fig:Country_dendrogram}. In Figures \ref{fig:India_index} and \ref{fig:Brazil_index}, we display the time series for India and Brazil, respectively. These are both characterized by a period of stability from 2010-2016; the primary difference is that India exhibits more change points in the first decade, while Brazil is more stable during this time. Russia is displayed in Figure \ref{fig:Russia_index}, distinguished by a lengthy period of stability in the first decade followed by regularly spaced change points. Finally, China (\ref{fig:China_index}) is distinguished as the only country that lacks a change point in 2020 associated with COVID-19. Indeed, the Chinese market recovered more quickly from the impact of the pandemic than any other country \cite{Chinaguardian_2020}.

\begin{figure}
    \centering
    \includegraphics[width=0.95\textwidth]{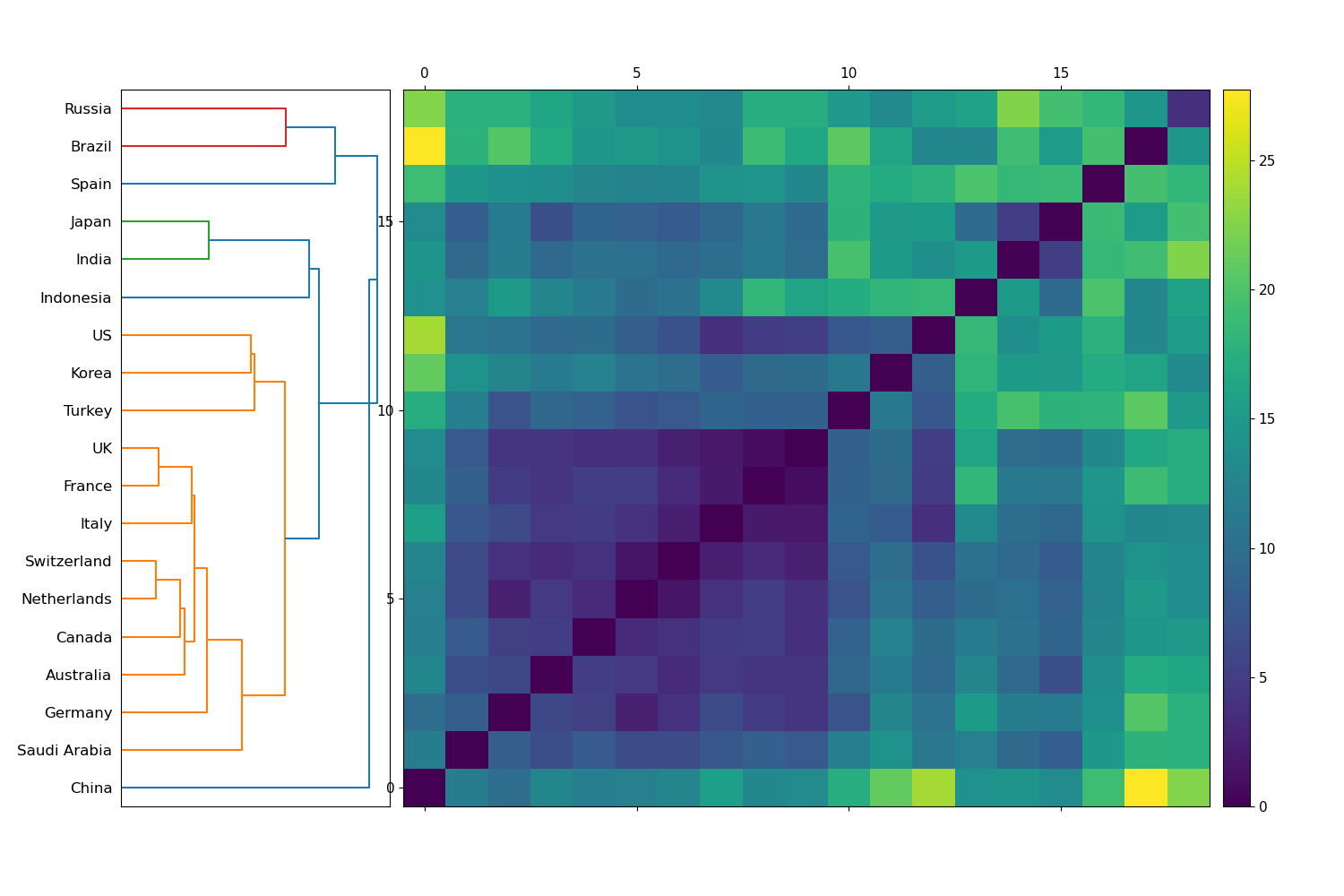} 
    \caption{Hierarchical clustering on $D^c$, the distance matrix between country indices' sets of change points with uncertainty. A strong subcluster of similarity consists of mostly European countries. Developing countries exhibit generally different erratic behavior profiles versus developed countries.}
    \label{fig:Country_dendrogram}
\end{figure}

\begin{figure}
    \centering
    \begin{subfigure}[b]{0.38\textwidth}
        \includegraphics[width=\textwidth]{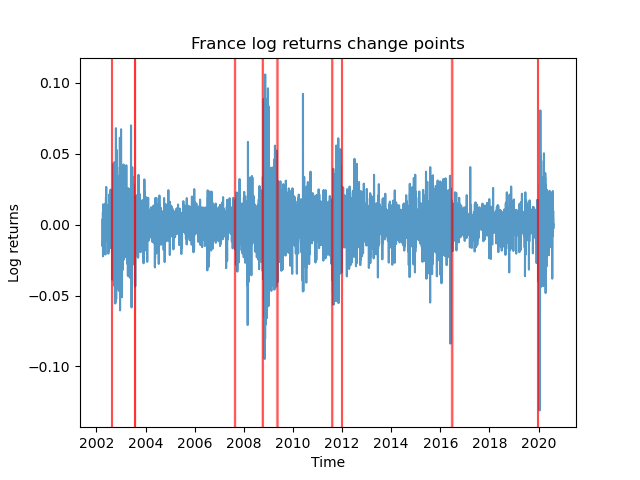}
        \caption{}
        \label{fig:France_index}
    \end{subfigure}
    \begin{subfigure}[b]{0.38\textwidth}
        \includegraphics[width=\textwidth]{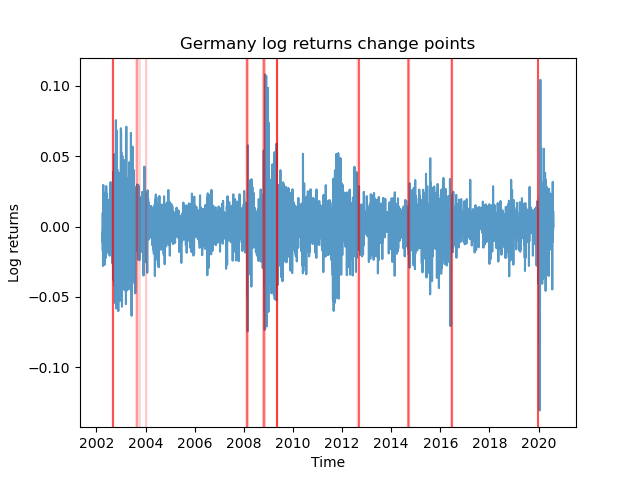}
        \caption{}
        \label{fig:Germany_index}
    \end{subfigure}
    \begin{subfigure}[b]{0.38\textwidth}
        \includegraphics[width=\textwidth]{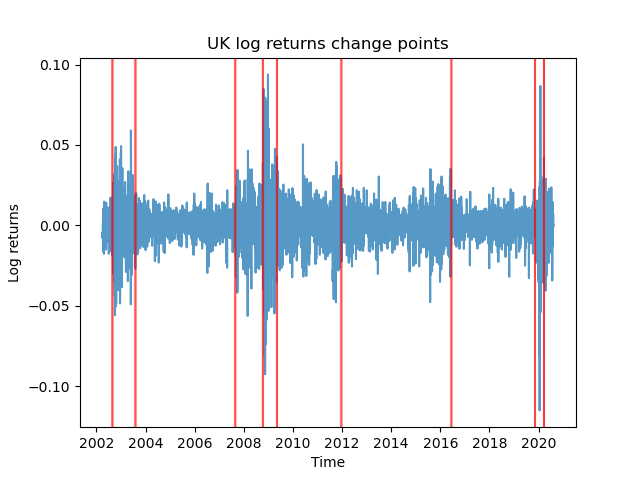}
        \caption{}
        \label{fig:UK_index}
    \end{subfigure}
    \begin{subfigure}[b]{0.38\textwidth}
        \includegraphics[width=\textwidth]{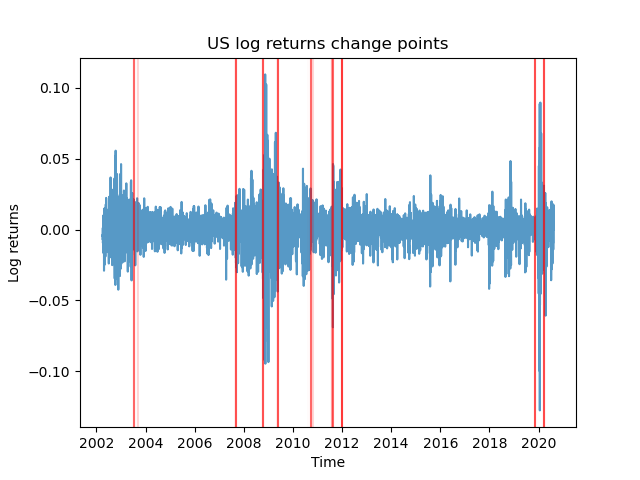}
        \caption{}
        \label{fig:US_index}
    \end{subfigure}
    \begin{subfigure}[b]{0.38\textwidth}
        \includegraphics[width=\textwidth]{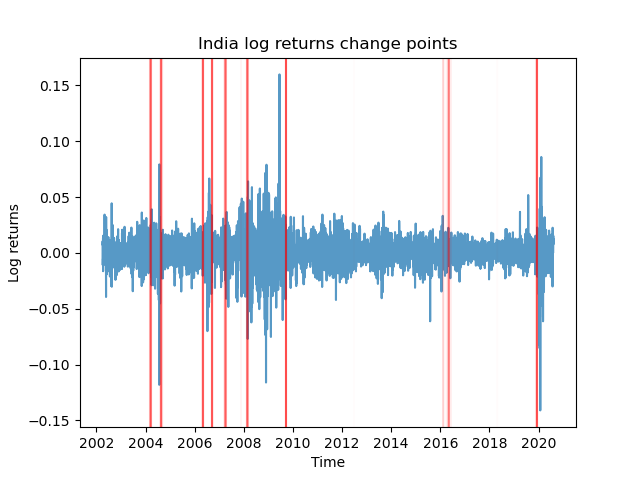}
        \caption{}
        \label{fig:India_index}
    \end{subfigure}
    \begin{subfigure}[b]{0.38\textwidth}
        \includegraphics[width=\textwidth]{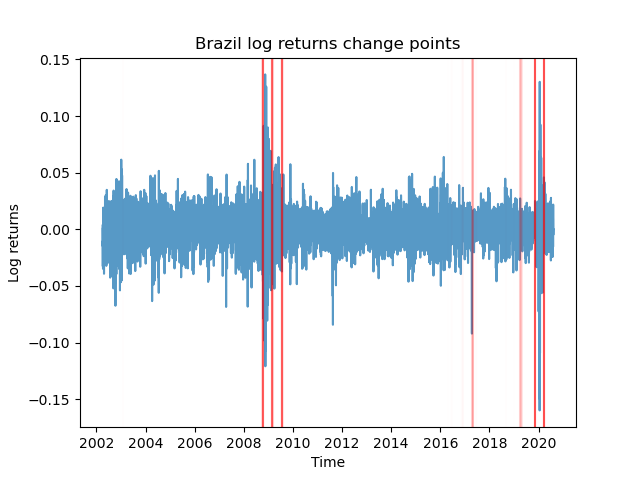}
        \caption{}
        \label{fig:Brazil_index}
    \end{subfigure}
    \begin{subfigure}[b]{0.38\textwidth}
        \includegraphics[width=\textwidth]{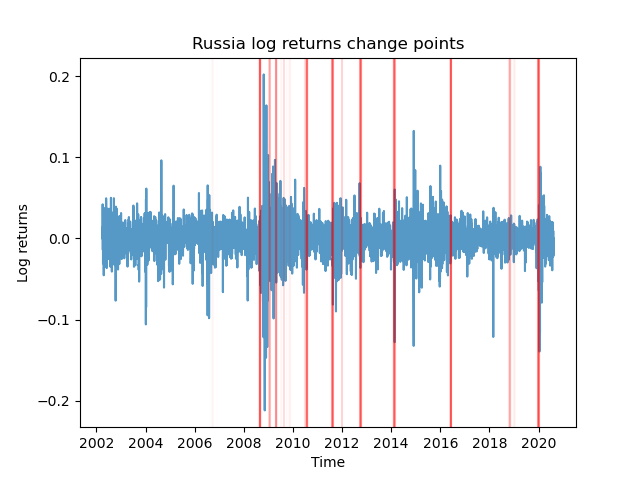}
        \caption{}
        \label{fig:Russia_index}
    \end{subfigure}
    \begin{subfigure}[b]{0.38\textwidth}
        \includegraphics[width=\textwidth]{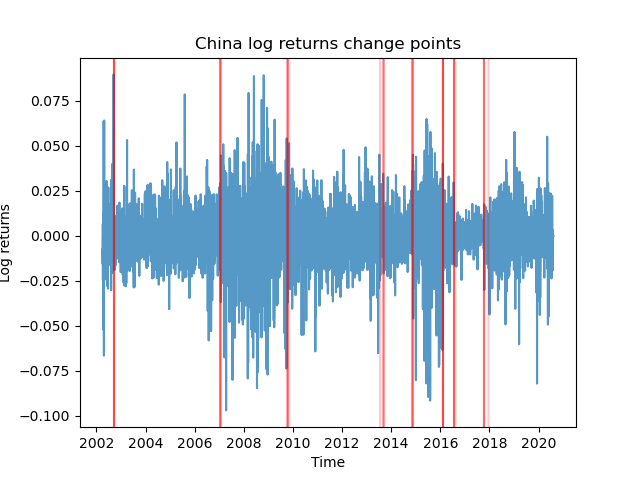}
        \caption{}
        \label{fig:China_index}
    \end{subfigure}
    \caption{Log return time series with detected change points for (a) France, (b) Germany, (c) the UK, (d) the US, (e) India, (f) Brazil, (g) Russia and (h) China. The transparency of the change point represents the value of the probability density function across its support interval. High similarity is observed between the erratic behavior profiles of France, Germany and the UK, with all three experiencing change points around the GFC, Brexit and COVID-19.}
    \label{fig:CountryTimeSeries}
\end{figure}

\subsection{Sector similarity}
Next, we apply the methodology of Section \ref{sec:MJWtheory} to the log return time series of 11 sector indices to produce an $11 \times 11$ distance matrix $D^s_{ij}$, and display hierarchical clustering of this matrix in Figure \ref{fig:Sector_dendrogram}. This dendrogram consists of a dominant cluster and two outliers. The primary cluster consists of communications, consumer discretionary, energy, financials, healthcare, information technology (IT), industrials, real estate and utilities. Of these,  consumer discretionary, energy, financials, healthcare, industrials and real estate comprise a subcluster of similarity. The two outlier sectors are consumer staples and materials. 

We examine select sectors and their change points in Figure \ref{fig:SectorTimeSeries}. In Figures \ref{fig:CDisc_index}, \ref{fig:Energy_index} and \ref{fig:Financials_index}, we display the log return time series for the consumer discretionary, energy and financial sectors, respectively. All three experience a change point at the start of 2020, and approximately evenly distributed breaks throughout the entire period. These three time series are also representative of the other members of the subcluster, that is, healthcare, industrials and real estate. Next, we display the log returns for communication services (\ref{fig:Communications_index}) and utilities (\ref{fig:Utilities_index}), which represent the remainder of the primary cluster in Figure \ref{fig:Sector_dendrogram}. Both of these are characterized by a period of stability in 2012-2019, much like that of the US time series (\ref{fig:US_index}). Finally, we display one of the outlier sectors, materials, in Figure \ref{fig:Materials_index}. This is characterized by relatively few change points early on. Broadly speaking, we can see that the sector time series have more similarity among themselves than the country indices. For example, every sector has a change point associated with the GFC and COVID-19, and sector change points are more evenly distributed throughout the entire period than for the countries. We shall observe this quantitatively as well.

\begin{figure}
    \centering
    \includegraphics[width=0.95\textwidth]{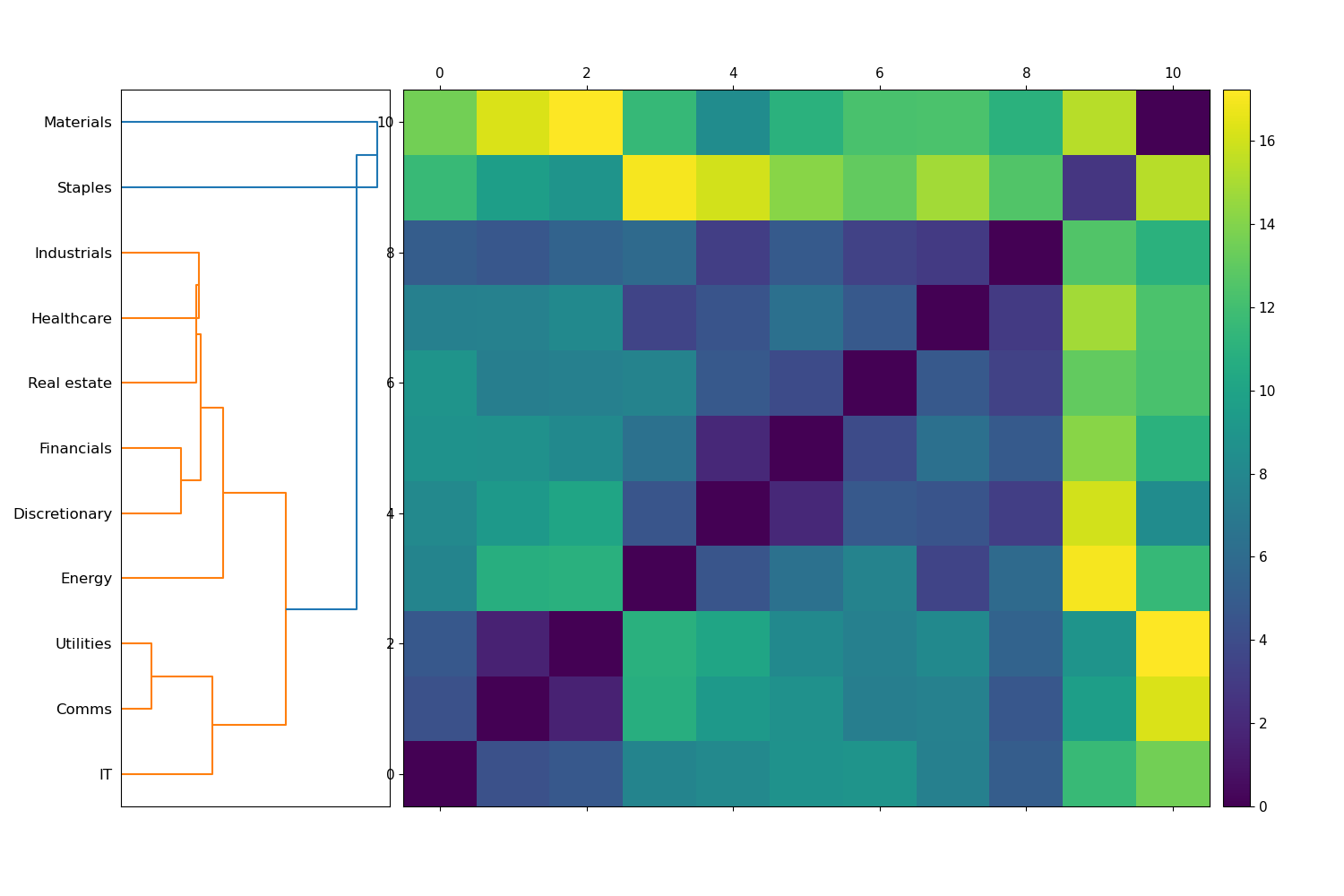}
    \caption{Hierarchical clustering on $D^s$, the distance matrix between sector indices' sets of change points with uncertainty. Materials and consumer staples (``Staples") are outliers. The most similarity is observed between consumer discretionary (``Discretionary"), energy, financials, healthcare, industrials and real estate. Communication services (``Comms"), IT and utilities show slightly different erratic behavior.}
    \label{fig:Sector_dendrogram}
\end{figure}

\begin{figure}
    \centering
    \begin{subfigure}[b]{0.48\textwidth}
        \includegraphics[width=\textwidth]{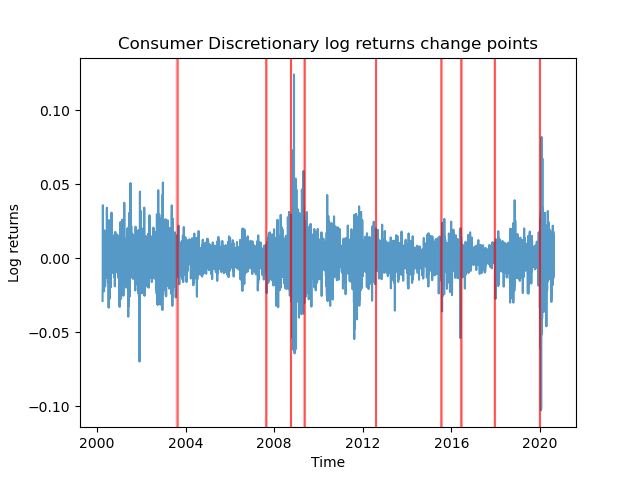}
        \caption{}
        \label{fig:CDisc_index}
    \end{subfigure}
    \begin{subfigure}[b]{0.48\textwidth}
        \includegraphics[width=\textwidth]{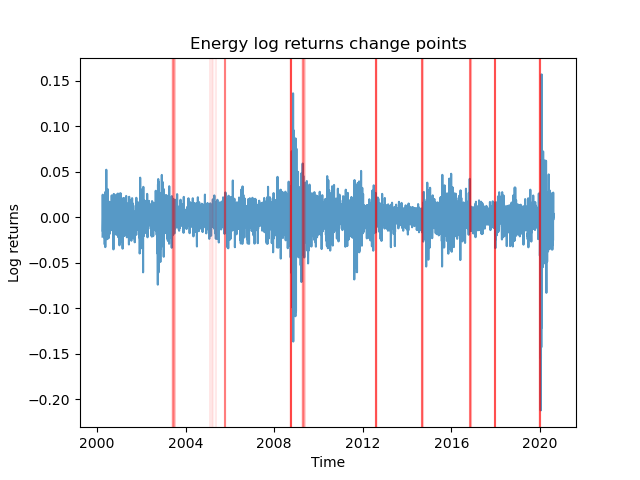}
        \caption{}
        \label{fig:Energy_index}
    \end{subfigure}
    \begin{subfigure}[b]{0.48\textwidth}
        \includegraphics[width=\textwidth]{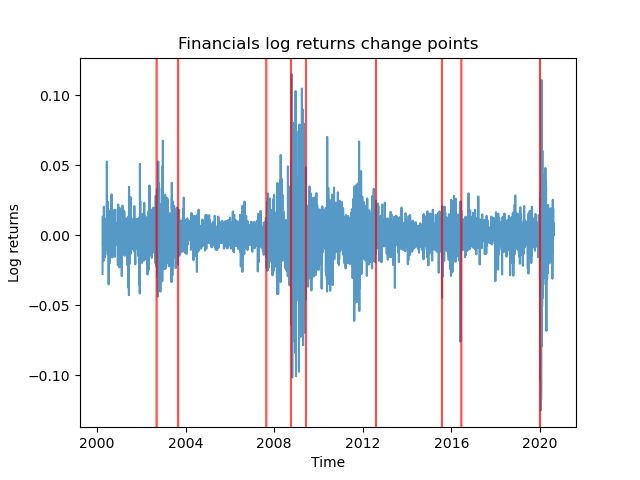}
        \caption{}
        \label{fig:Financials_index}
    \end{subfigure}
    \begin{subfigure}[b]{0.48\textwidth}
        \includegraphics[width=\textwidth]{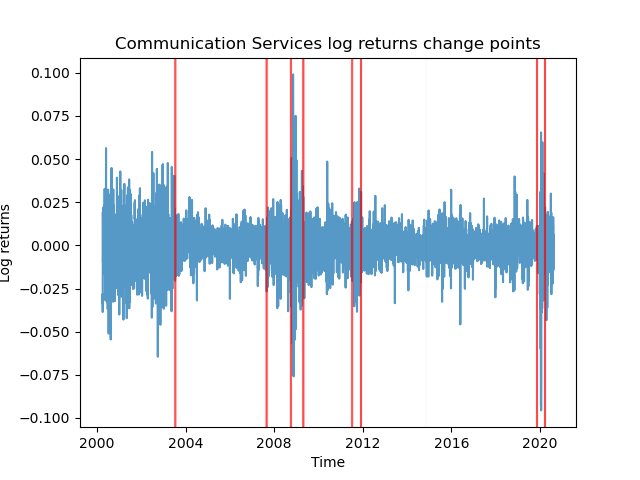}
        \caption{}
        \label{fig:Communications_index}
    \end{subfigure}
    \begin{subfigure}[b]{0.48\textwidth}
        \includegraphics[width=\textwidth]{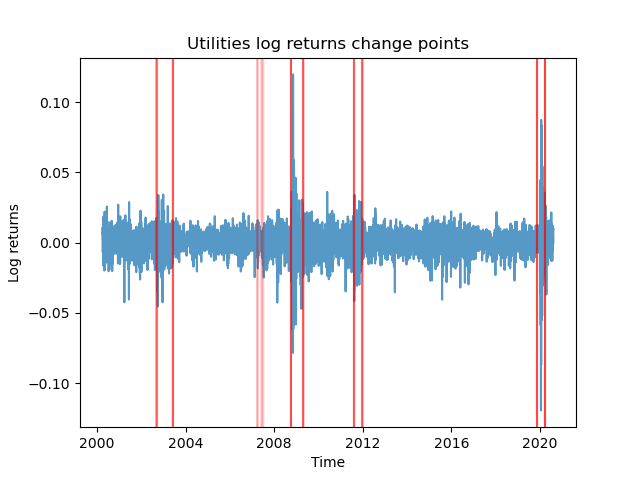}
        \caption{}
        \label{fig:Utilities_index}
    \end{subfigure}
    \begin{subfigure}[b]{0.48\textwidth}
        \includegraphics[width=\textwidth]{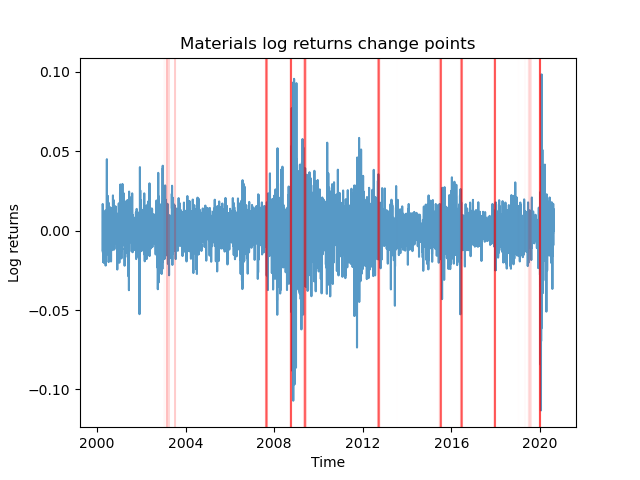}
        \caption{}
        \label{fig:Materials_index}
    \end{subfigure}
    \caption{Log return time series with detected change points for (a) consumer discretionary (b) energy, (c) financials, (d) communication services, (e) utilities and (f) materials. The transparency of the change point represents the value of the probability density function across its support interval. Every sector exhibits a change point in 2008-2009 around the GFC and in 2020 around COVID-19. Sector indices exhibit more similarity as a whole than country indices.}
    \label{fig:SectorTimeSeries}
\end{figure}

\begin{figure}
    \centering
    \begin{subfigure}[b]{0.48\textwidth}
        \includegraphics[width=\textwidth]{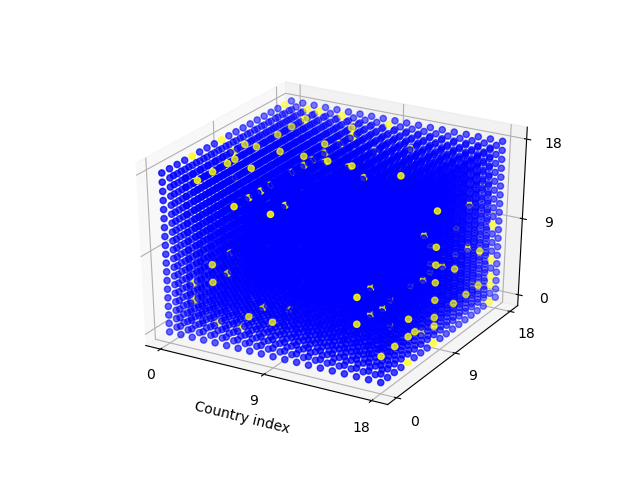}
        \caption{}
        \label{fig:Country_transitivity}
    \end{subfigure}
    \begin{subfigure}[b]{0.48\textwidth}
        \includegraphics[width=\textwidth]{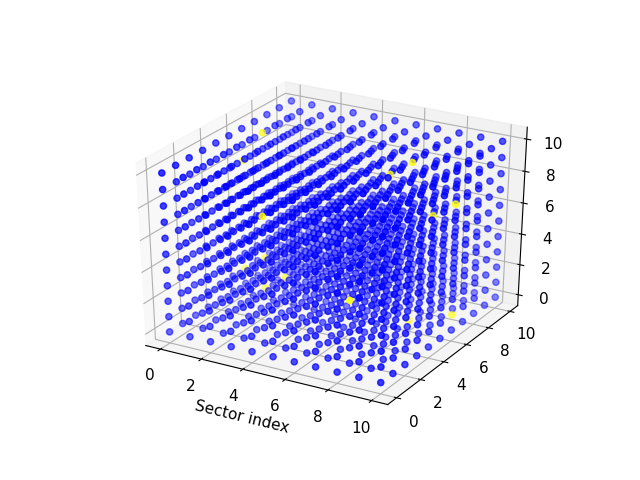}
        \caption{}
        \label{fig:Sector_transitivity}
    \end{subfigure}
    \caption{Triangle inequality test matrix, defined in (\ref{eq:transitivityanalysis}), for (a) countries and (b) sectors. A blue dot represents a triple that satisfies the triangle inequality, a yellow dot represents a failure of the triangle inequality by a factor of at most 2, and a red dot, of which there are none, indicates failure by a factor greater than 2. Although neither countries nor sectors fail the triangle inequality frequently or severely, countries do fail the triangle inequality slightly more frequently and severely than sectors. 2.45\% of candidate country triples fail, with an average fail of 1.10. 1.05\% of sector triples fail, with an average fail of 1.08.}
    \label{fig:Transitivity_tests}
\end{figure}

\subsection{Comparative analysis of the two collections}
In this section, we compare the properties of the two distance matrices between countries and sectors as a whole, analyzing different properties of their collective similarity. First, we analyze different matrix norms of $D^c$ and $D^s$. Let $D$ be an arbitrary symmetric $n \times n$ matrix  with diagonal entries equal to zero. Let
\begin{align}
    ||D||_1 &= \frac{1}{n(n-1)} \sum_{i,j=1}^n |D_{ij}|,\\
    ||D||_2 &= \left(\frac{1}{n(n-1)}\sum_{i,j=1}^n |D_{ij}|^2\right)^\frac12, \\
    ||D||_{op} &=\max_{v \in \mathbb{R}^n - \{0\}} \frac{||Dv||}{||v||} = \max \{|\lambda|: \lambda \text{ is an eigenvalue of } D \}.
\end{align}
These are referred to as the $L^1, L^2$ and \emph{operator} norms of the matrix $D$. We have rescaled the $L^1$ and $L^2$ norms by the number of non-zero elements of $D$, so we may compare matrices of different sizes. The operator norm does not require rescaling, as for example, the scalar matrix $kI_n$ has operator norm equal to $|k|$ regardless of $n$. The equality of the operator norm \cite{RudinFA} with the greatest eigenvalue holds as $D$ is symmetric, hence diagonalizable; this is known as the spectral theorem \cite{Axler}.

We record the different norms of $D^c$ and $D^s$ in Table \ref{tab:Distance_matrices_analysis}. As described in Section \ref{sec:MJWtheory}, the distance matrices are normalized by the slightly differing period lengths of the time series of countries and sectors (19 and 21 years, respectively). There, we confirm what we qualitatively observed in the above sections: the country distance matrix has consistently greater norms than the sectors, even after normalization by the number of elements and period length. This quantitatively shows greater discrepancy in distances between change points among countries than sectors. That is, the erratic behavior profiles are more similar among sectors than among countries. 

\subsection{Comparative analysis with existing distance measures}
\label{sec:comparativeanalysis}

In this section, we compare our results obtained from our methodology with several classical distances and similarity measures. We consider five measures between time series $x(t)$ and $y(t)$ over $t=1,...,T$:

\begin{enumerate}
    \item cosine similarity, defined as 
    \begin{align}
        \frac{\sum_{t=1}^T x(t)y(t)}{(\sum_{t=1}^T x(t)y(t))^\frac12 (\sum_{t=1}^T x(t)y(t))^\frac12 };
    \end{align}
    \item (Pearson) correlation, defined as the cosine similarity between $x(t)-\bar{x}$ and $y(t)-\bar{y}$;
    \item the Manhattan ($L^1$) distance \cite{Cantrell2000}, defined as    $\sum_{t=1}^T |x(t) - y(t)|$;
    \item the Euclidean ($L^2$) distance, defined as  $(\sum_{t=1}^T |x(t) - y(t)|^2)^\frac12$;
    \item the Chebyshev ($L^\infty$) distance \cite{Cantrell2000}, defined as $\max_t |x(t) - y(t)|$.
\end{enumerate}
In \ref{app:distances}, we display the ten heat maps obtained by using the five classical measures above to collections of both countries and sectors. Broadly, none of the observations concerning the erratic behavior of country and sector indices discussed previously in Section \ref{sec:results} can be obtained. We believe this demonstrates the necessity of using a new methodology to obtain inference on the similarity of sector or country indices with respect to their change points, taking uncertainty into account.

\subsection{Failure of the triangle inequality}

Finally, we analyze the failure of the triangle inequality when computing the semi-metric $d^1_{MJW}$ among the elements of each collection. As discussed in Section \ref{sec:MJWtheory}, the semi-metric does not satisfy the triangle inequality in general, so we propose an empirical method to test the triangle inequality among triples of elements in Section \ref{sec:transitivity analysis}. We form the matrix in (\ref{eq:transitivityanalysis}) for each collection, and display this in Figure \ref{fig:Transitivity_tests} for countries and sectors. The results highlight that neither sectors nor countries fail the triangle inequality frequently or severely, but triangle inequality failures are more frequent and slightly worse between country structural breaks. 1.05\% of sector triples fail the triangle inequality, with an average fail ratio of 1.08, while 2.45\% of country distances fail the triangle inequality with an average fail ratio of 1.10. We also record this in Table \ref{tab:Distance_matrices_analysis}.

\begin{table}[ht]
\begin{center}
\begin{tabular}{ |p{2.75cm}||p{1.75cm}|p{1.75cm}|}
 \hline
 \multicolumn{3}{|c|}{Comparative distance matrix analysis} \\
 \hline
  & Countries & Sectors \\
 \hline
  $L^1$ norm & 11.54 & 8.60 \\
 $L^2$ norm & 12.68 & 9.49 \\
  Operator norm & 221.1 & 90.8 \\
 \% of fails & 2.45\% & 1.05\% \\
 Average fail & 1.10 & 1.08 \\
\hline
\end{tabular}
\caption{Matrix norms and transitivity analysis results for $D^c$ and $D^s$, distance matrices between country and sector indices, respectively. We see quantitatively that country indices exhibit greater overall discrepancy in sets of change points than sectors.}
\label{tab:Distance_matrices_analysis}
\end{center}
\end{table}

\section{Conclusion}
\label{sec:conclusion}

This paper proposes a new method for measuring similarity in erratic behavior among a collection of time series. Our mathematical framework defines sets with uncertainty, introduces a new class of semi-metrics between them, and combines this with a suitable change point algorithm to quantify discrepancy between time series based on their change points, taking uncertainty into account. Our change point detection algorithm is judiciously chosen to perform well on dependent data, and to record the uncertainty around change points. Although our semi-metric may violate the triangle inequality, we include an empirical investigation and show that in all applications of the paper, the triangle inequality is violated infrequently and mildly. Thus, the semi-metric still respects a transitive property of closeness, as discussed in Section \ref{sec:MJWtheory}.

Applying this methodology to collections of country and sector indices over 20 years, we obtain immediate and visible insight into the structure of these collections with respect to their sets of change points, which are representatives of the indices' erratic behavior. We see clear similarity between the developed countries of our collection, particularly European countries. Examining these countries individually, we confirm similarity in several key change points, including the GFC, COVID-19, and Brexit in the case of European countries. Our methodology is also able to quickly identify outlier countries, such as China. Examining China reveals a clearly anomalous feature: no change point associated with COVID-19.

Applying our methodology to sectors and analyzing the resulting distance matrix compared to that of countries also produces interesting results. We show quantitatively that there is greater collective similarity among the erratic behavior profiles of sectors than countries. This can be qualitatively observed by examining the individual sector time series, where change points are consistently observed to be fairly equidistributed over time. In addition, every sector exhibits a change point around both the GFC and COVID-19. This observation has meaningful implications for investors, especially with regards to portfolio risk. Since there is greater similarity in sector erratic behaviors compared with countries, this suggests that diversification benefits are more substantial to investors diversifying with respect to countries (as opposed to sectors), as there is greater potential for reduction in change point variance.

The findings of the entire paper have several interpretations regarding modern trends in portfolio management. The metric introduced could accompany traditional risk management tools, and provide a framework for measuring distance between country and sector erratic behavior profiles. Our metric may capture dynamics and associations that are not explicitly identified in traditional measures such as correlation. The experiments in this paper may provide an indication for potential use cases for investors concerned with asset allocation across countries and sectors. Future work could use analysis of the erratic behavior profiles of different time series, such as individual equities, to provide an alternative or modification of traditional portfolio diversification.

In summary, this paper introduces a new measure to study similarity between time series' change points (as determined by a Bayesian change point algorithm), while capturing the uncertainty around such erratic shifts. We present promising findings on country and sector indices, with substantial insights into the structure of the two collections and the differences therein. Our methodology and findings pair well with existing statistical techniques, with the identification of frequently observed change points inspiring a closer examination of the associated crises. This closer analysis identifies relationships between erratic behavior profiles, correlations, dynamics and trajectories, with noteworthy implications for financial practitioners and new understandings of risk management during crises.

\appendix
\section{Bayesian change point detection algorithm}
\label{appendix:RJMCMCsampling}

In this section, we describe the Bayesian change point detection algorithm that we use in the paper. Specifically, we use a reversible jump Markov chain Monte Carlo (RJMCMC) algorithm \cite{Rosen2017} that continually updates a partitioning of the time series into $m$ segments, where $m$ also changes with time. The sampling scheme produces a distribution over $m$ and a set of $m$ change points. As described in Section \ref{sec:MJWtheory}, we conclude by selecting the maximal likely $m_0$; then our set of change points varies with uncertainty conditional on $m_0$.

Our methodology detects change points based on changes in the time-varying power spectrum \cite{Rosen2009}. This addresses many of the limitations of traditional change point detection algorithms to handle dependent data. Prior work \cite{Rosen2017} has demonstrated that the methodology applied in this paper appropriately detects changes in piecewise autoregressive processes, highlighting the algorithm's ability to partition dependent data. This is crucial when analyzing financial time series, where rich autocorrelation structure has been observed in the literature.

We follow Rosen et al. \cite{Rosen2012} in our implementation of the RJMCMC sampling scheme. We denote a varying partition of the time series by $\bs{\xi}_{m} = (\xi_{0,m},...,\xi_{m,m})$; these are our $m$ change points (excluding $\xi_{m,m}$, which by convention is always the final time point). The algorithm requires the consideration of a vector of \emph{amplitude parameters} $\bs{\tau}_{m}^{2} = (\tau_{1,m}^{2},...,\tau_{m,m}^{2})'$ and \emph{regression coefficients} $\bs{\beta}_{m} = (\bs{\beta'}_{1,m},...,\bs{\beta'}_{m,m})$ that we wish to estimate, for the $j$th component within a partition of $m$ segments, $j=1,...,m.$ For notational convenience, $\bs{\beta}_{j,m}, j=1,...,m,$ is assumed to include the first entry, $\alpha_{0j,m}.$ In the proceeding sections, superscripts $c$ and $p$ refer to current and proposed value in the RJMCMC sampling scheme, respectively.

First, we describe the \textbf{between-model moves}, where the number of change points $m$ may be changed. Let $\bs{\theta}_{m} = (\bs{\xi}'_{m}, \bs{\tau}^{2'}_{m}, \bs{\beta'}_{m})$ be the model parameters at some point in the RJMCMC sampling scheme and assume that the chain starts at $(m^c, \bs{\theta}_{m^c}^{c})$. The algorithm proposes a move to $(m^p, \bs{\theta}_{m^p}^p)$, by drawing $(m^p, \bs{\theta}_{m^p}^p)$ from a proposal distribution  $q(m^p, \bs{\theta}_{m^p}^p|m^c, \bs{\theta}_{m^c}^c)$. This draw is accepted with probability
\begin{equation*}
    \alpha = \text{ min } \Bigg\{1, \frac{p(m^{p}, \bs{\theta}_{m^p}^{p}|\bs{x}) q(m^c, \bs{\theta}_{m^c}^{c}|m^p, \bs{\theta}_{m^p}^{p})}
    {p(m^{c}, \bs{\theta}_{m^c}^{c}|\bs{x}) q(m^p, \bs{\theta}_{m^p}^{p}|m^c, \bs{\theta}_{m^c}^{c})} \Bigg\},   
\end{equation*}
where $p(\cdot)$ is a target distribution, namely the product of the likelihood and the prior. The target and proposal distributions vary based on the type of move taken in the sampling scheme. The distribution $q(m^p, \bs{\theta}_{m^{p}}^{p}| m^c, \bs{\theta}_{m^{c}}^{c})$ is described as follows:
\begin{align*}
    q(m^p, \bs{\theta}_{m^p}^{p}|m^c, \bs{\theta}_{m^{c}}^{c}) = q(m^p|m^c)  q(\bs{\theta}_{m^p}^{p}| m^p, m^c, \bs{\theta}_{m^c}^{c}) 
    = q(m^p|m^c)  q(\bs{\xi^p_{m^p}}, \bs{\tau}_{m^p}^{2p}, \bs{\beta}_{m^p}^{p} | m^p, m^c, \bs{\theta}_{m^c}^c) \\
    = q(m^p|m^c)  q(\bs{\xi}_{m^p}^p|m^p, m^c, \bs{\theta}_{m^c}^c)  q(\bs{\tau}_{m^p}^{2p}|\bs{\xi}_{m^p}^p, m^p, m^c, \bs{\theta}_{m^c}^c)  q(\bs{\beta}_{m^p}^{p}|\bs{\tau}_{m^p}^{2p}, \bs{\xi}_{m^p}^p, m^p, m^c, \bs{\theta}_{m^c}^c).
\end{align*}
To draw $(m^p, \bs{\theta}_{m^p}^p)$, one must first draw $m^p$, and then $\bs{\xi}_{m^p}^p$, $\tau_{m^p}^{2p}, \text{ and } \bs{\beta}_{m^p}^{p}$. The number of segments $m^p$ is drawn from the proposal distribution $q(m^p|m^c)$. Let $M$ be the maximum number of segments and $m^{c}_{2,\text{min}}$ be the number of current segments containing at least $2  t_{min}$ data points. The proposal is as follows:
\begin{align*}
    q(m^p = k | m^c)= \left\{
                \begin{array}{ll}
                  1/2 \text{ if } k = m^c - 1, m^c + 1 \text{ and } m^{c} \neq 1, M, m_{2,\text{min}}^c \neq 0\\
                  1 \text{ if } k = m^{c}-1 \text{ and } m^{c} = M \text{ or } m_{2,\text{min}}^{c} = 0 \\
                  1 \text{ if } k = m^{c} + 1 \text{ and } m^{c} = 1
                \end{array}
              \right.
\end{align*}
Conditional on the proposed number of change points $m^p$, a new partition $\bs{\xi}_{m^p}^{p}$, a new vector of covariance amplitude parameters $\bs{\tau}_{m^p}^{2p}$, and a new vector of regression coefficients, $\bs{\beta}_{m^p}^{p}$ are proposed. $\tau^{2}$ is described as a smoothing parameter \cite{Rosen2012} or amplitude parameter.

Now, we describe the process of the \textbf{birth} of new segments, where the number of proposed change points increases. Suppose that $m^p = m^c + 1$. As defined before, a time series partition,
    \begin{align*}
    \bs{\xi}^{p}_{m^p} = (\xi^c_{0,m^c},...,\xi^c_{k^{*}-1,m^c},\xi_{k^{*},m^{p}}^{p}, \xi_{k^{*},m^{c}}^{c},...,\xi_{m^c,m^c}^{c})    
    \end{align*}
    is drawn from the proposal distribution $q(\bs{\xi}_{m^p}^{p}|m^p, m^c, \bs{\theta}_{m^c}^c)$. The algorithm first proposes a partition by selecting a random segment $j = k^{*}$ to split. Then, a random point $t^{*}$ within the segment $j=k^{*}$ is selected to be the proposed partition point. This is subject to a constraint, $\xi_{k^{*}-1, m^c}^{c} + t_{\text{min}} \leq t* \leq \xi_{k^{*},m^c}^c - t_{\text{min}}$. The proposal distribution is then computed as follows:
    \begin{align*}
        q(\xi_{j,m^p}^{p} = t^{*} | m^p, m^c, \bs{\xi}_{m^c}^{c}) =  p(j=k^{*} | m^p, m^c, \bs{\xi}_{m^c}^{c}); \\
         p(\xi_{k^{*}, m^p}^{p} = t^{*} | j=k^{*}, m^p, m^c, \bs{\xi}_{m^c}^c)
        =  \frac{1}{m_{2 \text{min}}^{c}(n_{k^{*}, m^c} - 2t_{\text{min}}+1)}.
    \end{align*}
Then, the vector of amplitude parameters
    \begin{align*}
        \tau_{m^p}^{2p} = (\tau_{1,m^c}^{2c},...,\tau_{k^{*}-1,m^c}^{2c},  \tau_{k^{*},m^p}^{2p}, \tau_{k^{*}+1,m^p}^{2p}, \tau_{k^{*}+1,m^c}^{2c},...,\tau_{m^c, m^c}^{2c})
    \end{align*}
    is drawn from the proposal distribution $q(\bs{\tau}_{m^p}^{2p}|m^p, \bs{\xi}_{m^p}^p, m^c, \bs{\theta}_{m^c}^c) = q(\bs{\tau}_{m^p}^{2p}|m^p, \bs{\tau}_{m^c}^{2c}).$ The algorithm is based on the RJMCMC algorithm of \cite{GREEN1995}. This draws from a uniform distribution $u \sim U[0,1]$ and defines $\tau_{k^{*}, m^p}^{2p}$ and $\tau_{k^{*}+1, m^p}^{2p}$ in terms of $u$ and $\tau_{k^{*}, m^c}^{2c}$ as follows:
    \begin{align}
    \label{eq:birth1}
        \tau_{k^{*, m^p}}^{2p} =   \frac{u}{1-u}\tau_{k^{*}, m^c}^{2c}; \\
        \tau_{k^{*}+1, m^p}^{2p} =   \frac{1-u}{u}\tau_{k^{*}, m^c}^{2c}.
        \label{eq:birth2}
    \end{align}
Then, the vector of coefficients
    \begin{align*}
        \bs{\beta}_{m^p}^p = (\bs{\beta}_{1,m^c}^{c},...,\bs{\beta}_{k^{*}-1,m^c}^{c},  \bs{\beta}_{k^{*}, m^p}^p, \bs{\beta}_{k^{*}+1,m^p}^{p}, \bs{\beta}_{k^{*}+1,m^c}^{c},...,\bs{\beta}_{m^c,m^c}^{c})
    \end{align*}
    is drawn from the proposal distribution  $q(\bs{\beta}_{m^p}^p|\bs{\tau}_{m^p}^{2p},\bs{\xi}_{m^p}^{2p},m^p, m^c, \bs{\theta}_{m^c}^c) = q(\bs{\beta}_{m^p}^{p}|\bs{\tau}_{m^p}^{2p}, \bs{\xi}_{m^p}^p, m^p)$. The vectors $\bs{\beta}_{k^{*}, m^p}^p$ and $\bs{\beta}_{k^{*}+1, m^p}^p$ are drawn from Gaussian approximations to the respective posterior conditional distributions $p(\bs{\beta}_{k^{*}, m^p}^p|\bs{x}_{k^{*}}^p, \tau_{k^{*}, m^p}^{2p}, m^p)$ and  $p(\bs{\beta}_{k^{*}+1, m^p}^{p}|\bs{x}_{k^{*}+1}^p, \tau_{k^{*}+1, m^p}^{2p}, m^p)$, respectively. Here, $\bs{x}_{k^{*}}^p$ and $\bs{x}_{k^{*}+1}^p$ refer to the subsets of the time series across segments $k^{*}$ and $k^{*}+1$, respectively. Then, $\bs{\xi}_{m^p}^p$  determines $\bs{x_{*}}^p = (\bs{x}_{k^{*}}^{p'}, \bs{x}_{k^{*}+1}^{p'})'$. We provide an example for illustration: the coefficient $\bs{\beta}_{k^{*}, m^p}^{p}$ is drawn from the Gaussian distribution $N(\bs{\beta}_{k^{*}}^{\text{max}}, \Sigma_{k^{*}}^{\text{max}})$, where  
    \begin{align*}
     \bs{\beta}_{k^{*}}^{\text{max}}=    \argmax_{\bs{\beta}_{k^{*}, m^p}^{p}} p(\bs{\beta}^p_{k^{*}, m^p}|\bs{x}_{k^{*}}^p, \tau_{k^{*}, m^p}^{2p}, m^p)
    \end{align*} and 
    \begin{align*}
      \Sigma_{k^{*}}^{\text{max}} =-\Bigg \{ \pdv{\log p(\bs{\beta}_{k^{*}, m^p}^p | \bs{x}_{k^{*}}^p, \tau_{k^{*}, m^p}^{2p}, m^p)}{\bs{\beta}_{k^{*}, m^p}^{p}}{\bs{\beta}_{k^{*}, m^p}^{p'}} \Bigg|_{\bs{\beta}_{k^{*}, m^p}^{p} = \bs{\beta}_{k^{*}}^{\text{max}}}  \Bigg \}_.^{-1}
    \end{align*}
    For the birth move (that is the increase in $m$), the probability of acceptance is $\alpha = \min\{1,A\}$, where $A$ is equal to
    \begin{align*}
        \Bigg| \pdv{(\tau_{k^{*}, m^p}^{2p}, \tau_{k^{*}+1, m^p}^{2p})}{(\tau_{k^{*}, m^c, u}^{2c})}  \Bigg|\frac{p(\bs{\theta}_{m^p}^p|\bs{x}, m^p) p(\bs{\theta}_{m^p}^p|m^p)p(m^p)}{p(\bs{\theta}_{m^p}^p|\bs{x}, m^p) p(\bs{\theta}_{m^c}^c|m^c)p(m^c)}  \frac{p(m^{c}|m^p)p(\bs{\beta}_{k^{*}, m^c}^{c})}{p(m^p|m^c)p(\xi_{k^{*}, m^p}^{m^p}|m^p, m^c) p(u) p(\bs{\beta}^p_{k^{*}, m^p})p(\bs{\beta}_{k^{*}+1,m^{p}}^{p})}.  
    \end{align*}
    Above, $p(u) = 1, 0 \leq u \leq 1,$ while $p(\bs{\beta}_{k^{*}, m^p}^{p})$ and $p(\bs{\beta}_{k^{*}+1, m^p}^{p})$ are the density functions of Gaussian proposal distributions $N(\bs{\beta}_{k^{*}}^{\text{max}}, \Sigma_{k^{*}}^{\text{max}})$ and $N(\bs{\beta}_{k^{*}+1}^{\text{max}}, \Sigma_{k^{*}+1}^{\text{max}})$, respectively. The above Jacobian is computed as
    \begin{align*}
        \bigg| \pdv{(\tau_{k^{*}, m^p}^{2p}, \tau_{k^{*}+1, m^p}^{2p})}{(\tau^{2c}_{k^{*}, m^c}, u)} \bigg| = \frac{2 \tau_{k^{*}m^c}^{2c}}{u(1-u)} = 2(\tau_{k^{*}, m^p}^{p} + \tau_{k^{*}+1, m^p}^{p})^{2}.
    \end{align*}

Now, we move on to describe the process of the \textbf{death} of new segments, that is, where the number of proposed change points decreases, or $m^p = m^c - 1$. A time series partition
\begin{align*}
\bs{\xi}_{m^p}^{p} = (\xi_{0,m^c}^{c},...,\xi_{k^{*}-1,m^c}^{c}, \xi_{k^{*}+1,m^c}^{c},...,\xi_{m^c,m^c}^{c}),
\end{align*}
is proposed by randomly selecting a single change point from $m^c - 1$ candidates, and removing it. The change point selected for removal is denoted $j=k^{*}$. There are $m^c -1$ possible change points available for removal among the $m^c$ segments currently in existence. The proposal may choose each change point with equal probability, that is,
\begin{align*}
    q(\xi_{j, m^p}^p|m^p, m^c, \bs{\xi}_{m^c}^c) = \frac{1}{m^c - 1}.
\end{align*}
The updated vector of amplitude parameters 
\begin{align*}
\bs{\tau}_{m^p}^{2p} = (\tau_{1, m^c}^{2c},...,\tau_{k^{*}-1,m^c}^{2c},\tau_{k^{*},m^p}^{2c}, \tau_{k^{*}+2,m^c}^{2c},...,\tau_{m^c,m^c}^{2c})    
\end{align*}
 is then drawn from the proposal distribution  $q(\bs{\tau}_{m^p}^{2p}|m^p, \bs{\xi}_{m^p}^p, m^c, \bs{\theta}_{m^c}^{c}) = q(\bs{\tau}_{m^p}^{2p}| m^p, \bs{\tau}_{m^c}^{2c})$. Then, one amplitude parameter $\tau_{k^{*}, m^p}^{2p}$ is formed from two candidate amplitude parameters, $\tau_{k^{*},m^c}^{2c}$ and $\tau_{k^{*}+1,m^c}^{2c}$,. by reversing the equations (\ref{eq:birth1}) and (\ref{eq:birth2}). Specifically,
\begin{align*}
    \tau_{k^{*}, m^p}^{2p} = \sqrt{\tau_{k^{*}, m^{c}}^{2c} \tau_{k^{*}+1, m^c}^{2c}}.
\end{align*}
Finally, the updated vector of regression coefficients,
\begin{align*}
    \bs{\beta}_{m^p}^{p} = (\beta_{1,m^c}^{c},...,\beta_{k^{*}-1,m^c}^{c}, \beta_{k^{*}, m^p}^{p}, \beta_{k^{*}+2,m^c}^{c},...,\beta_{m^c,m^c}^{c})
\end{align*}
is drawn from the proposal distribution  $q(\bs{\beta}_{m^p}^p|\bs{\tau}_{m^p}^{2p}, \bs{\xi}_{m^p}^p, m^p, m^c, \theta_{m^c}^c) = q(\bs{\beta}_{m^p}^{p}|\bs{\tau}_{m^p}^{2p}, \bs{\xi}_{m^p}^p, m^p)$. The vector of regression coefficients is drawn from a Gaussian approximation to the posterior distribution $p(\beta_{k^{*},m^p}|\bs{x}, \tau_{k^{*}, m^p}^{2p}, \bs{\xi}^p_{m^p}, m^p)$ with the same procedure for the vector of coefficients in the birth step. The probability of acceptance is the inverse of the corresponding birth step. If the move is accepted then the following updates of the current values occur: $m^c=m^p$ and $\bs{\theta}_{m^c}^c = \bs{\theta}_{m^p}^{p}$.

Finally, we describe the \textbf{within-model moves:}
henceforth, the number of change points $m$ is fixed and notation describing the dependence on the number of change points is removed. There are two parts of a within-model move. First, the change points may be relocated (with the same number), and conditional on the relocation, the basis function coefficients are updated. The steps are jointly accepted or rejected with a Metropolis-Hastings step, and the amplitude parameters are updated within a separate Gibbs sampling step. 

We assume the chain is located at $\bs{\theta}^{c} = (\bs{\xi}^{c}, \bs{\beta}^{c})$. The proposed move $\bs{\theta}^p = (\bs{\xi}^p, \bs{\beta}^p)$ is as follows: first, a change point $\xi_{k^{*}}$ is selected for relocation from $m-1$ candidate change points. Next, a position within the interval $[\xi_{k^{*}-1}, \xi_{k^{*}+1}]$ is chosen, subject to the fact that the new location is at least $t_{\text{min}}$ data points away from $\xi_{k^{*}-1}$ and $\xi_{k^{*}+1}$, so that
\begin{align*}
    \Pr(\xi^p_{k^{*}}=t) = \Pr (j=k^{*})  \Pr (\xi_{k^{*}}^{p}=t|j=k^{*}), 
\end{align*}
where $\Pr(j=k^{*}) = (m-1)^{-1}$. A mixture distribution for $\Pr(\xi_{k^{*}}^p=t|j=k^{*})$ is constructed to explore the space efficiently, so
\begin{align*}
    \Pr(\xi_{k^{*}}^{p}=t|j=k^{*}) =  \pi q_1 (\xi_{k^{*}}^p = t| \xi_{k^{*}}^{c}) + (1-\pi) q_2 (\xi_{k^{*}}^p=t|\xi_{k^{*}}^c),
\end{align*}
where $q_1(\xi_{k^{*}}^p = t| \xi_{k^{*}}^c) = (n_{k^{*}} + n_{k^{*}+1}-2t_{\text{min}} + 1)^{-1}, \xi_{k^{*}-1} + t_{\text{min}} \leq t \leq \xi_{k^{*}+1} - t_{\text{min}}$ and 

\begin{align*}
    q_2(\xi_{k^{*}}^p = t|\xi_{k^{*}}^{c})= 
\left\{
            \begin{array}{ll}
              0 \text{ if } |t-\xi^c_{k^{*}}| > 1;  \\
            1/3 \text{ if } |t-\xi^c_{k^{*}}| \leq 1, n_{k^{*}} \neq t_{\text{min}} \text{ and } n_{k^{*}+1} \neq t_{\text{min}};  \\ 
            1/2 \text{ if } t-\xi_{k^{*}}^{c} \leq 1, n_{k^{*}} = t_{\text{min}} \text{ and } n_{k^{*}+1} \neq t_{\text{min}}; \\
            1/2 \text{ if } \xi_{k^{*}}^{c} - t \leq 1, n_{k^{*}} \neq t_{\text{min}} \text{ and } n_{k^{*}+1} = t_{\text{min}}; \\
            1 \text{ if } t = \xi_{k^{*}}^{c}, n_{k^{*}} = t_{\text{min}} \text{ and } n_{k^{*}+1} = t_{\text{min}}.
            \end{array}
          \right.
\end{align*}
The support of $q_1$ has $n_{k^{*}} + n_{k^{*}+1} - 2t_{\text{min}} + 1$ points while $q_2$ has at most three. The term $q_2$ alone would result in a high acceptance rate for the Metropolis-Hastings step, but it would explore the parameter space slowly. The $q_1$ component allows for larger steps, and produces a compromise between a high acceptance rate and a thorough exploration of the parameter space. 

 Then, $\bs{\beta^{p}_{j}}, j=k^{*}, k^{*}+1$ is drawn from an approximation to $\prod^{k^{*}+1}_{j=k^{*}} p(\bs{\beta}_j|\bs{x}_j^p, \tau_j^{2})$, following the corresponding step in the between-model move. The proposal distribution, evaluated at $\bs{\beta}^{p}_j, j=k^{*}, k^{*}=1$, is
\begin{align*}
    q(\bs{\beta}_{*}^{p}|\bs{x}_{*}^p, \bs{\tau}_{*}^{2}) = \prod^{k^{*}+1}_{j=k^{*}} q(\bs{\beta}_{j}^p|\bs{x}_j^p, \tau_j^{2}),
\end{align*}
where $\bs{\beta}_{*}^p = (\bs{\beta}^{p'}_{k^{*}}, \bs{\beta}^{p'}_{k^{*}+1})'$ and $\bs{\tau}_{*}^{2} = (\tau^{2}_{k^{*}}, \tau^{2}_{k^{*}+1})'$. This proposal distribution is also evaluated at current values of $\bs{\beta}_{*}^{c} = (\beta^{c'}_{k^{*}}, \beta^{c'}_{k^{*}+1})'$. $\beta_{*}^p$ is then accepted with probability
\begin{equation*}
    \alpha = \min \Bigg\{ 1, \frac{p(\bs{x}_{*}^p|\bs{\beta}_{*}^{p}) p(\bs{\beta}_{*}^{p}|\bs{\tau}_{*}^{2}) q(\bs{\beta}_{*}^{c}|\bs{x}^{c}_{*}, \bs{\tau}_{*}^{2})} {p(\bs{x}_{*}^c|\bs{\beta}_{*}^{c}) p(\bs{\beta}_{*}^{c}|\bs{\tau}_{*}^{2}) q(\bs{\beta}_{*}^{p}|\bs{x}^{p}_{*}, \bs{\tau}_{*}^{2})} \Bigg\},
\end{equation*}
where $\bs{x}_{*}^{c} = (\bs{x}^{c'}_{k^{*}},\bs{x}^{c'}_{k^{*}+1})$. When the draw is accepted, we update the partition and regression coefficients $(\xi^{c}_{k^{*}}, \beta_{*}^{c}) = (\xi^{p}_{k^{*}}, \beta_{*}^{p})$.
Finally, we draw $\tau^{2p}$ from 
\begin{align*}
    p(\tau_{*}^{2}|\bs{\beta}_{*}) = \prod^{k^{*}+1}_{j=k^{*}} p(\tau_j^{2}|\beta_j).
\end{align*}
This is a Gibbs sampling step, and as such the draw is accepted with probability 1. 

This concludes the description of the algorithm. For our purposes, we do not need the vectors of amplitude parameters or regression coefficients, just the distribution of the change points.

\section{Details of synthetic autoregressive experiment}
\label{app:autoregressive}

First, we explicitly detail the processes governing the six synthetic time series discussed in Section \ref{sec:validation}. In what follows, $\epsilon_t$ are independent and identically distributed $N(0,0.1)$ Gaussian noise terms.

\begin{align}
x^{(1)}_{t}=\begin{cases}
			0.9 x^{(1)}_{t-1} + \epsilon_{t}, & \text{1 $\leq t \leq$ 200};\\
            1.5 x^{(1)}_{t-1} - 0.75 x^{(1)}_{t-2} + \epsilon_{t}, & \text{201  $\leq t \leq$ 500}; \\
            0.9 x^{(1)}_{t-1} + \epsilon_{t}, & \text{501  $\leq t \leq$ 700}; \\
            0.9 x^{(1)}_{t-1} - 0.8 x^{(1)}_{t-2 }+ \epsilon_{t}, & \text{701  $\leq t \leq$ 900}; \\
            - 0.9 x^{(1)}_{t-1} + \epsilon_{t}, & \text{901  $\leq t \leq$ 1100}; \\
            0.9 x^{(1)}_{t-1} + \epsilon_{t}, & \text{1101  $\leq t \leq$ 1300}; \\
            0.9 x^{(1)}_{t-1} - 0.8 x^{(1)}_{t-2} + \epsilon_{t}, & \text{1301  $\leq t \leq$ 1500}; \\
		 \end{cases}
\end{align}

\begin{align}
x^{(2)}_{t}=\begin{cases}
			1.5 x^{(2)}_{t-1} - 0.75 x^{(2)}_{t-2} + \epsilon_{t}, & \text{1  $\leq t \leq$ 195};\\
			0.9 x^{(2)}_{t-1} - 0.8 x^{(2)}_{t-2} + \epsilon_{t}, & \text{196  $\leq t \leq$ 500};\\
			1.5 x^{(2)}_{t-1} - 0.75 x^{(2)}_{t-2} + \epsilon_{t}, & \text{501  $\leq t \leq$ 690};\\
			0.9 x^{(2)}_{t-1} - 0.8 x^{(2)}_{t-2} + \epsilon_{t}, & \text{691  $\leq t \leq$ 900};\\
			-0.9 x^{(2)}_{t-1} + \epsilon_{t}, & \text{901  $\leq t \leq$ 1110};\\
			1.5 x^{(2)}_{t-1} - 0.75 x^{(2)}_{t-2} + \epsilon_{t}, & \text{1111  $\leq t \leq$ 1300};\\
			0.9 x^{(2)}_{t-1} - 0.8 x^{(2)}_{t-2} + \epsilon_{t}, & \text{1301  $\leq t \leq$ 1500};\\
		 \end{cases}
\end{align}

\begin{align}
x^{(3)}_{t}=\begin{cases}
			0.9 x^{(3)}_{t-1} - 0.8 x^{(3)}_{t-2} + \epsilon_{t}, & \text{1  $\leq t \leq$ 190}; \\
			0.9 x^{(3)}_{t-1} + \epsilon_{t}, & \text{191  $\leq t \leq$ 500}; \\
			-0.9 x^{(3)}_{t-1} + \epsilon_{t}, & \text{501  $\leq t \leq$ 685}; \\
			0.9 x^{(3)}_{t-1} - 0.8 x^{(3)}_{t-2} + \epsilon_{t}, & \text{686  $\leq t \leq$ 900}; \\
			-0.9 x^{(3)}_{t-1} + \epsilon_{t}, & \text{901  $\leq t \leq$ 1105}; \\
			0.9 x^{(3)}_{t-1} - 0.8 x^{(3)}_{t-2} + \epsilon_{t}, & \text{1106  $\leq t \leq$ 1300}; \\
			1.5 x^{(3)}_{t-1} - 0.75 x^{(3)}_{t-2} + \epsilon_{t}, & \text{1301  $\leq t \leq$ 1500}; \\
		 \end{cases}
\end{align}

\begin{align}
x^{(4)}_{t}=\begin{cases}
			-0.9 x^{(4)}_{t-1} + \epsilon_{t}, & \text{1  $\leq t \leq$ 190};\\
			0.9 x^{(4)}_{t-1} + \epsilon_{t}, & \text{191  $\leq t \leq$ 500};\\
			0.9 x^{(4)}_{t-1} - 0.8 x^{(4)}_{t-2} + \epsilon_{t}, & \text{501  $\leq t \leq$ 685};\\
			1.5 x^{(4)}_{t-1} - 0.75 x^{(4)}_{t-2} + \epsilon_{t}, & \text{686  $\leq t \leq$ 900};\\
			-0.9 x^{(4)}_{t-1} + \epsilon_{t}, & \text{901  $\leq t \leq$ 1105};\\
			0.9 x^{(4)}_{t-1} - 0.8 x^{(4)}_{t-2} + \epsilon_{t}, & \text{1106  $\leq t \leq$ 1300};\\
			-0.9 x^{(4)}_{t-1} + \epsilon_{t}, & \text{1301  $\leq t \leq$ 1500};\\
		 \end{cases}
\end{align}

\begin{align}
x^{(5)}_{t}=\begin{cases}
			-0.9 x^{(5)}_{t-1} + \epsilon_{t}, & \text{1  $\leq t \leq$ 750}; \\
			1.5 x^{(5)}_{t-1} - 0.75 x^{(5)}_{t-2} + \epsilon_{t}, & \text{751  $\leq t \leq$ 1500}; \\
		 \end{cases}
\end{align}

\begin{align}
x^{(6)}_{t}=\begin{cases}
			0.9 x^{(6)}_{t-1} - 0.8 x^{(6)}_{t-2} + \epsilon_{t}, & \text{1  $\leq t \leq$ 750}; \\
			-0.9 x^{(6)}_{t-1} + \epsilon_{t}, & \text{751  $\leq t \leq$ 1500}; \\
		 \end{cases}
\end{align}

Next, we plot the above six time series in Figure \ref{fig:ARTS}, as well as the determined sets of change points with uncertainty. While not perfect, we believe this validates our methodology as determining difficult-to-detect change points in time series with a high level of autocorrelation.

Finally, we describe our sensitivity analysis regarding small changes in autoregressive parameters. We form four time series of length $T=2000$, each with a subtle break at $t=1000$. Each of the time series is of the same form:
\begin{align}
y^{(i)}_{t}=\begin{cases}
			0.9 y^{(i)}_{t-1}  + \epsilon_{t}, & \text{1  $\leq t \leq$ 1000}; \\
			c_i y^{(i)}_{t-1} + \epsilon_{t}, & \text{1001  $\leq t \leq$ 2000}, \\
		 \end{cases}
\end{align}
where $c_i$ is a specified constant for each time series $i=1,...,4$. We select constants $c_1=0.5, c_2=0.6, c_3=0.7, c_4=0.8$. Applying our change point methodology, change points near $t=1000$ are successfully detected for $c_1=0.5$ and $c_2=0.6$ but not for the even smaller changes $c_3=0.7$ and $c_4=0.8$ relative to 0.9. For $c_1$ and $c_2$, change points are estimated at $t=1009$ and $t=977$, respectively. Thus, we may observe successful detection of change points in all but the slightest changes of autoregressive parameters.

\begin{figure}
    \centering
    \begin{subfigure}[b]{0.48\textwidth}
        \includegraphics[width=\textwidth]{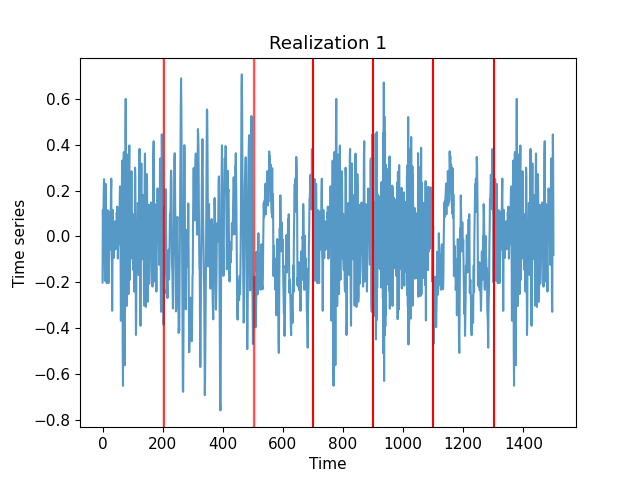}
        \caption{}
        \label{fig:AR_1}
    \end{subfigure}
    \begin{subfigure}[b]{0.48\textwidth}
        \includegraphics[width=\textwidth]{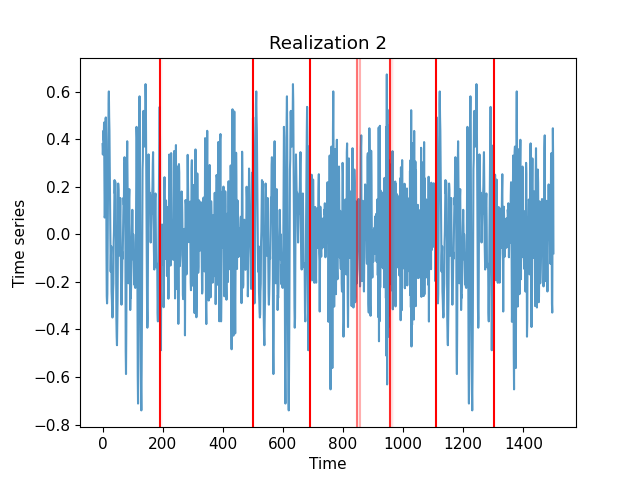}
        \caption{}
        \label{fig:AR_2}
    \end{subfigure}
    \begin{subfigure}[b]{0.48\textwidth}
        \includegraphics[width=\textwidth]{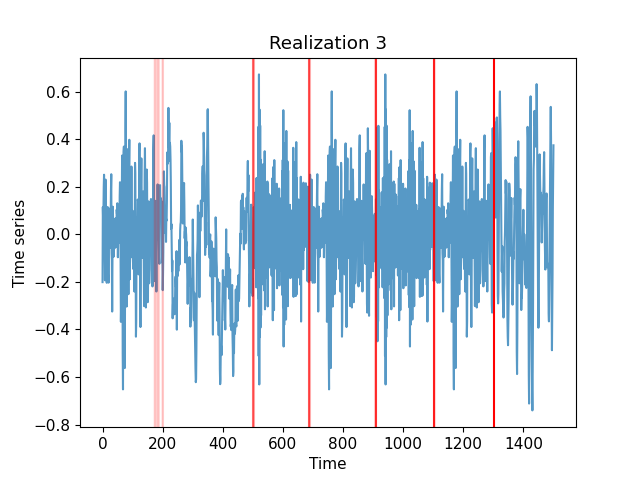}
        \caption{}
        \label{fig:AR_3}
    \end{subfigure}
    \begin{subfigure}[b]{0.48\textwidth}
        \includegraphics[width=\textwidth]{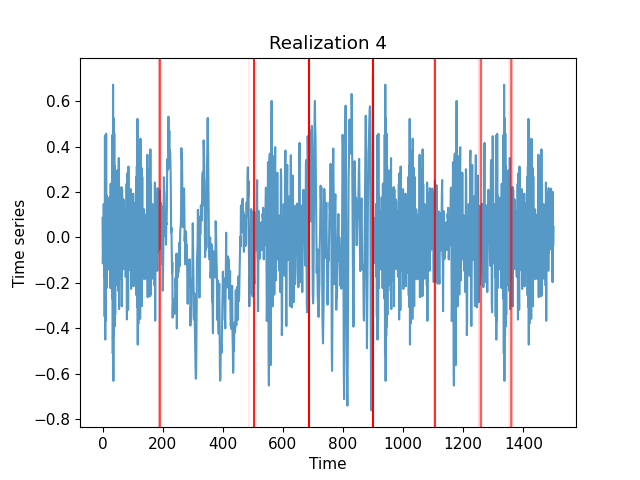}
        \caption{}
        \label{fig:AR_4}
    \end{subfigure}
    \begin{subfigure}[b]{0.48\textwidth}
        \includegraphics[width=\textwidth]{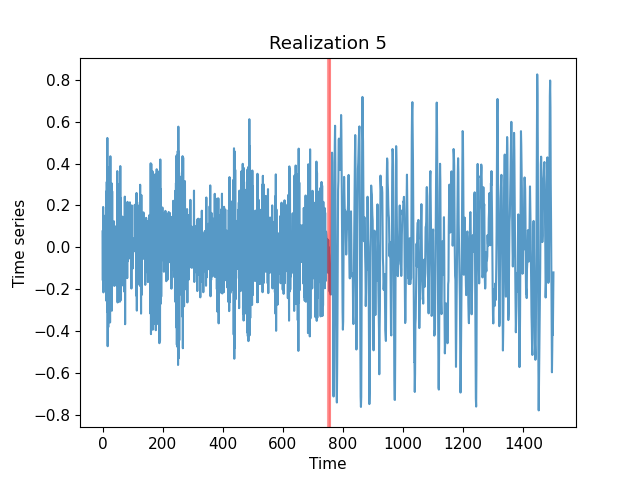}
        \caption{}
        \label{fig:AR_5}
    \end{subfigure}
    \begin{subfigure}[b]{0.48\textwidth}
        \includegraphics[width=\textwidth]{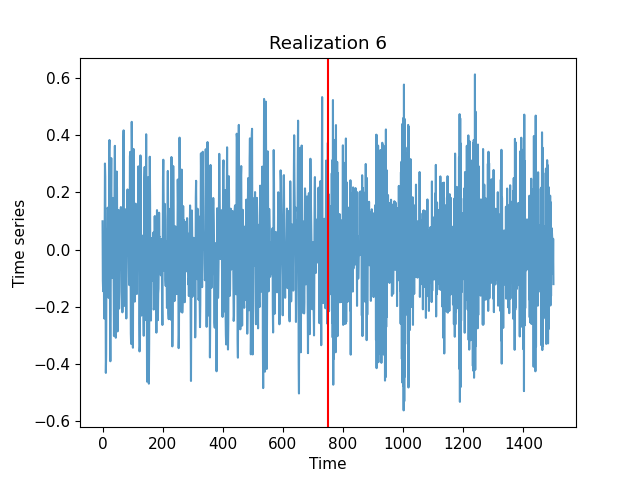}
        \caption{}
        \label{fig:AR_6}
    \end{subfigure}
    \caption{Six synthetic time series built from piecewise autoregressive processes, together with detected change points. The transparency of the change point represents the value of the probability density function across its support interval.}
    \label{fig:ARTS}
\end{figure}

\section{Comparison to existing distance measures}
\label{app:distances}

In this brief section, we plot ten heat maps obtained by applying the five classical measures defined in Section \ref{sec:comparativeanalysis} to the collections of country and sector indices. For a consistent comparison, we make a small adjustment to the latter three metrics (the Manhattan, Euclidean and Chebyshev distances).

The cosine similarity and correlation take values between $-1$ and $1$, with 1 indicating the greatest possible similarity. The latter three metrics indicate greatest similarity with a value of zero. For ease of comparison, we linearly adjust the latter three metrics as follows: given a $n \times n$ distance matrix $D$, let its associated affinity matrix be defined as
\begin{align}
    A_{ij}=1 - \frac{D_{ij}}{\max D}, i,j=1,...,n.
\end{align}
Then an affinity value of 1 between two time series indicates greatest possible similarity.

Having performed this transformation, we plot heat maps of the cosine similarity and correlation matrices as well as the affinity matrices associated to the Manhattan, Euclidean and Chebyshev distances in Figure \ref{fig:altdistances}. Essentially none of the findings reported in Section \ref{sec:results} can be drawn from these figures. We believe this demonstrates the utility of the methodology of the paper.

\begin{figure}
    \centering
    \begin{subfigure}[b]{0.32\textwidth}
        \includegraphics[width=\textwidth]{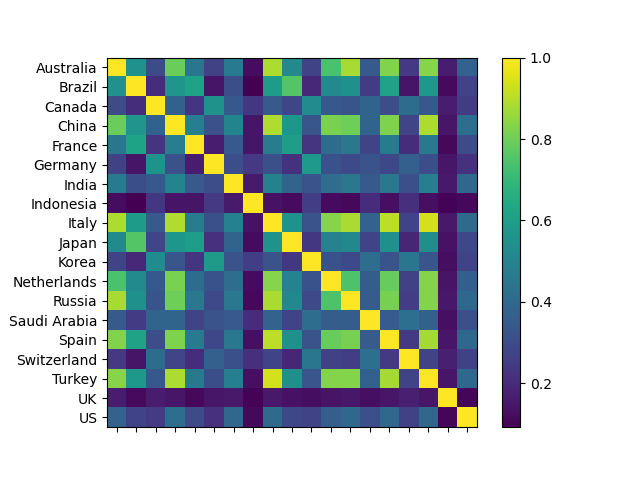}
        \caption{}
        \label{fig:corr_countries}
    \end{subfigure}
    \begin{subfigure}[b]{0.32\textwidth}
        \includegraphics[width=\textwidth]{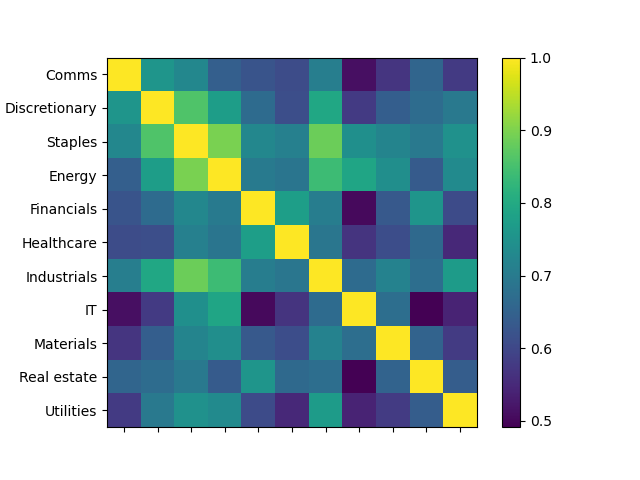}
        \caption{}
        \label{fig:corr_sectors}
    \end{subfigure}
    \begin{subfigure}[b]{0.32\textwidth}
        \includegraphics[width=\textwidth]{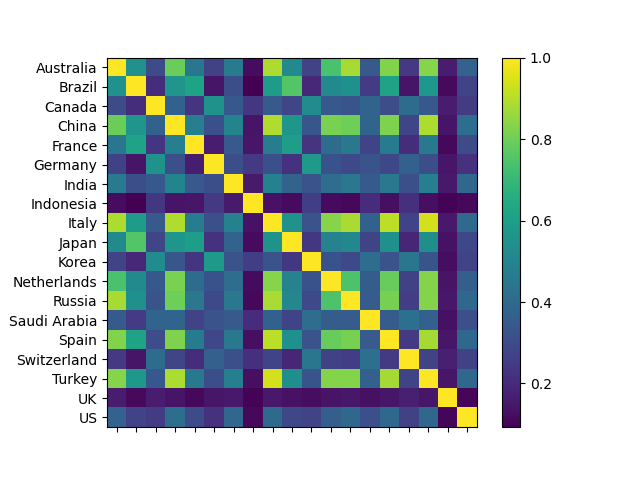}
        \caption{}
        \label{fig:cos_countries}
    \end{subfigure}
    \begin{subfigure}[b]{0.32\textwidth}
        \includegraphics[width=\textwidth]{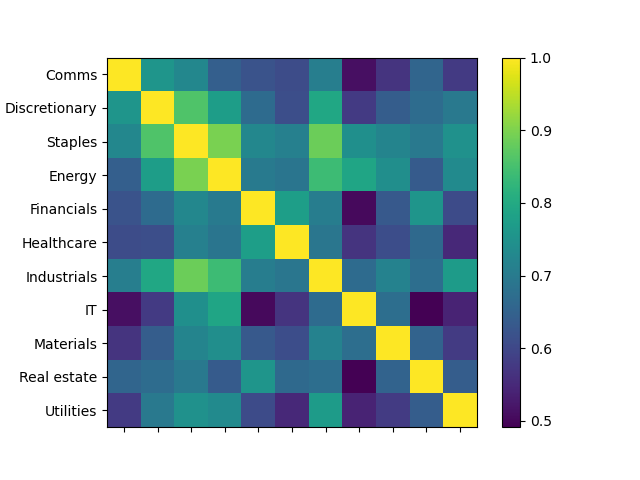}
        \caption{}
        \label{fig:cos_sectors}
    \end{subfigure}
    \begin{subfigure}[b]{0.32\textwidth}
        \includegraphics[width=\textwidth]{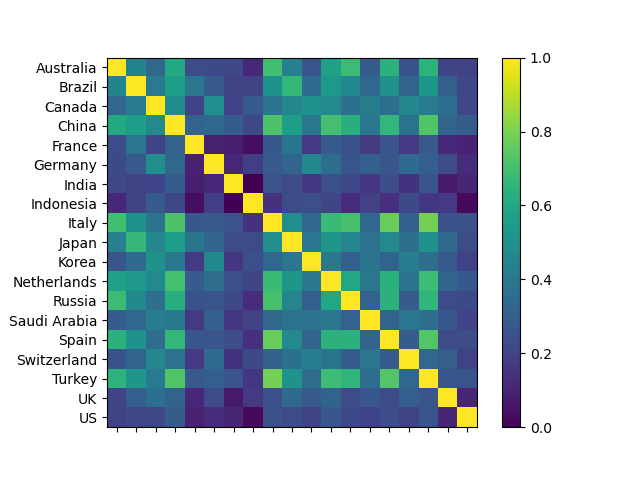}
        \caption{}
        \label{fig:man_countries}
    \end{subfigure}
    \begin{subfigure}[b]{0.32\textwidth}
        \includegraphics[width=\textwidth]{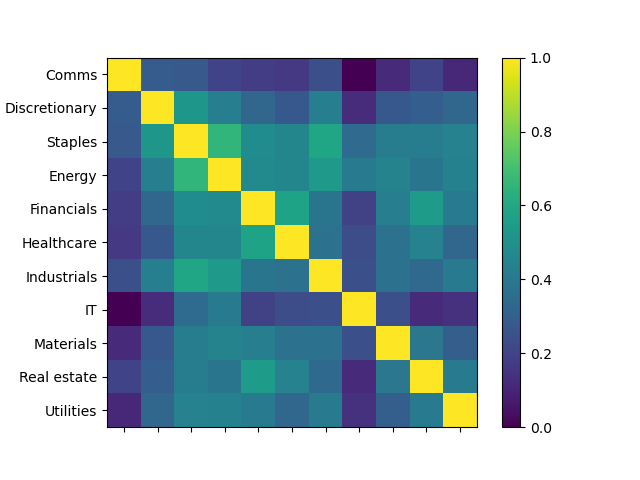}
        \caption{}
        \label{fig:man_sectors}
    \end{subfigure}
    \begin{subfigure}[b]{0.32\textwidth}
        \includegraphics[width=\textwidth]{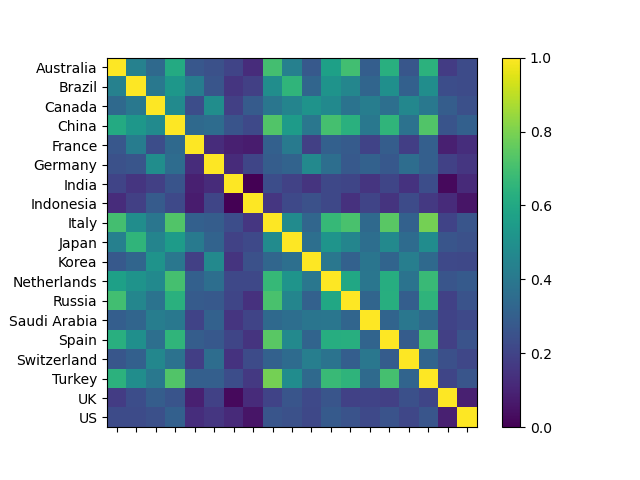}
        \caption{}
        \label{fig:euc_countries}
    \end{subfigure}
    \begin{subfigure}[b]{0.32\textwidth}
        \includegraphics[width=\textwidth]{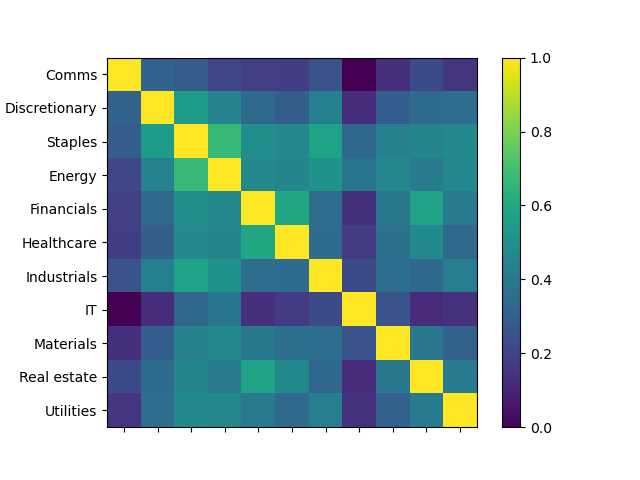}
        \caption{}
        \label{fig:euc_sectors}
    \end{subfigure}
    \begin{subfigure}[b]{0.32\textwidth}
        \includegraphics[width=\textwidth]{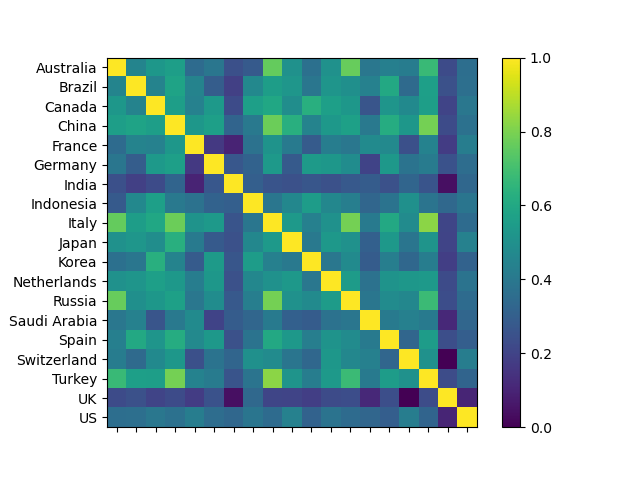}
        \caption{}
        \label{fig:Cheb_countries}
    \end{subfigure}
    \begin{subfigure}[b]{0.32\textwidth}
        \includegraphics[width=\textwidth]{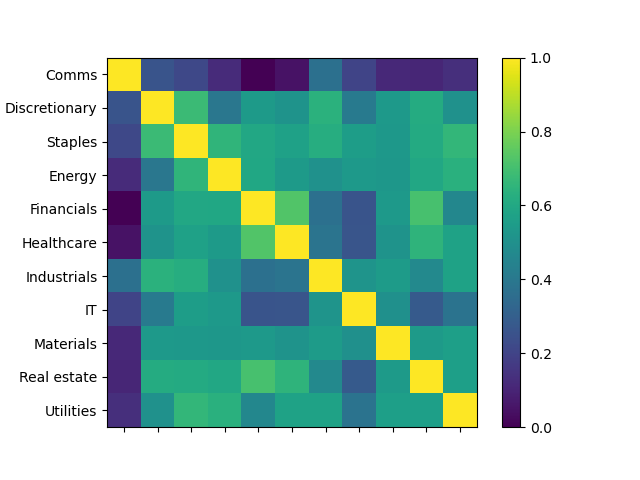}
        \caption{}
        \label{fig:Cheb_sectors}
    \end{subfigure}
    \caption{Heat maps of classical similarity measures between country and sector indices. We display the cosine similarity for (a) countries and (b) sectors as well as correlation for (c) countries and (d) sectors. Subsequently, we plot the affinity matrices associated to three metrics, the Manhattan metric in (e) and (f), the Euclidean metric in (g) and (h), and the Chebyshev metric in (i) and (j), each for countries and sectors, respectively.}
    \label{fig:altdistances}
\end{figure}

\bibliographystyle{iopart-num}
\bibliography{_references.bib}

\end{document}